\documentclass[format=acmsmall, nonacm=true]{acmart}
\usepackage{acm-ec-24-proc}
\usepackage{booktabs} 
\usepackage[ruled]{algorithm2e} 
\usepackage{natbib} 

\SetAlFnt{\small}
\SetAlCapFnt{\small}
\SetAlCapNameFnt{\small}
\SetAlCapHSkip{0pt}
\IncMargin{-\parindent}

\usepackage{mathtools}
\usepackage{bm}
\usepackage{enumitem}

\usepackage{tikz}
\usetikzlibrary{patterns}
\usetikzlibrary{patterns.meta}

\usepackage[capitalise,nameinlink,noabbrev]{cleveref}

\setcitestyle{authoryear}

\usepackage{multirow}
\usepackage{mdframed}

\theoremstyle{definition}

\theoremstyle{plain}

\usepackage{aliascnt}

\newtheorem{theorem}{Theorem}[section]

\newaliascnt{lemma}{theorem}
\newtheorem{lemma}[lemma]{Lemma}
\aliascntresetthe{lemma}

\newaliascnt{corollary}{theorem}
\newtheorem{corollary}[corollary]{Corollary}
\aliascntresetthe{corollary}

\newaliascnt{proposition}{theorem}

\aliascntresetthe{proposition}

\theoremstyle{remark}

\DeclareMathOperator{\opt}{OPT}

\def\fp/{\textup{\textsf{FP}}}
\def\p/{\textup{\textsf{P}}}
\def\np/{\textup{\textsf{NP}}}
\def\conp/{\textup{\textsf{co-NP}}}
\def\fnp/{\textup{\textsf{FNP}}}
\def\tfnp/{\textup{\textsf{TFNP}}}
\def\ptfnp/{\textup{\textsf{PTFNP}}}
\def\ppa/{\textup{\textsf{PPA}}}
\def\ppad/{\textup{\textsf{PPAD}}}
\def\ppads/{\textup{\textsf{PPADS}}}
\def\ppp/{\textup{\textsf{PPP}}}
\def\pwpp/{\textup{\textsf{PWPP}}}
\def\pls/{\textup{\textsf{PLS}}}
\def\cls/{\textup{\textsf{CLS}}}
\def\ppadpls/{\textup{$\textsf{PPAD} \cap \textsf{PLS}$}}
\def\ppapls/{\textup{$\textsf{PPA} \cap \textsf{PLS}$}}
\def\eopl/{\textup{\textsf{EOPL}}}
\def\sopl/{\textup{\textsf{SOPL}}}
\def\ueopl/{\textup{\textsf{UEOPL}}}
\def\fixp/{\textup{\textsf{FIXP}}}
\def\bu/{\textup{\textsf{BU}}}
\def\bbu/{\textup{\textsf{BBU}}}
\def\linearfixp/{\textup{\textsf{Linear-FIXP}}}
\def\pspace/{\textup{\textsf{PSPACE}}}

\providecommand{\ceil}[1]{\ensuremath{\left \lceil #1 \right \rceil }}

\def\pcircuit{\textup{\textsc{Pure-Circuit}}\xspace}

\newcommand{\garbo}{\ensuremath{\bot}\xspace}

\newcommand{\val}[1]{\boldsymbol{\mathrm{x}}[#1]}
\newcommand{\valonly}{\boldsymbol{\mathrm{x}}}

\newcommand{\PURE}{\textup{\textsf{PURIFY}}\xspace}
\newcommand{\PURIFY}{\PURE}

\newcommand{\NAND}{\textup{\textsf{NAND}}\xspace}
\newcommand{\NOT}{\textup{\textsf{NOT}}\xspace}

\newcommand{\eps}{\ensuremath{\varepsilon}\xspace}

\newcommand{\reals}{\ensuremath{\mathbb{R}}\xspace}

\newcommand{\refg}{\ensuremath{\mathsf{ref}}\xspace}
\newcommand{\refb}{\ensuremath{b_\text{ref}}\xspace}

\newcommand{\hlow}{\ensuremath{H_\text{low}}\xspace}
\newcommand{\hhigh}{\ensuremath{H_\text{high}}\xspace}
\newcommand{\llow}{\ensuremath{L_\text{low}}\xspace}
\newcommand{\lhigh}{\ensuremath{L_\text{high}}\xspace}
\newcommand{\hmax}{\ensuremath{H_\text{max}}\xspace}
\newcommand{\hmin}{\ensuremath{H_\text{min}}\xspace}
\newcommand{\tmin}{\ensuremath{\bar{t}}\xspace}

\newcommand{\pairs}{\ensuremath{w}\xspace}

\newcommand{\aux}{\ensuremath{\mathsf{aux}}\xspace}

\newcommand{\inverter}{\ensuremath{\mathsf{inverter}}\xspace}
\newcommand{\ing}{\ensuremath{\mathsf{in}}\xspace}
\newcommand{\inoneg}{\ensuremath{\mathsf{in1}}\xspace}
\newcommand{\intwog}{\ensuremath{\mathsf{in2}}\xspace}
\newcommand{\outg}{\ensuremath{\mathsf{out}}\xspace}
\newcommand{\outoneg}{\ensuremath{\mathsf{out1}}\xspace}
\newcommand{\outtwog}{\ensuremath{\mathsf{out2}}\xspace}
\newcommand{\midg}{\ensuremath{\mathsf{mid}}\xspace}

\newcommand{\rnot}{\ensuremath{r_\text{\NOT}}\xspace}
\newcommand{\rnand}{\ensuremath{r_\text{\NAND}}\xspace}
\newcommand{\prefmin}{\ensuremath{p^-_\refg}\xspace}
\newcommand{\prefmax}{\ensuremath{p^+_\refg}\xspace}



\acmYear{2024}\copyrightyear{2024}
\setcopyright{rightsretained}
\acmConference[EC '24]{Conference on Economics and Computation}{July 8--11, 2024}{New Haven, CT, USA}
\acmBooktitle{Conference on Economics and Computation (EC '24), July 8--11, 2024, New Haven, CT, USA}
\acmDOI{10.1145/3670865.3673533}
\acmISBN{979-8-4007-0704-9/24/07}

\title[Constant Inapproximability for Fisher Markets]{Constant Inapproximability for Fisher Markets}

\titlenote{A preliminary version of this work appeared at EC 2024~\cite{DeligkasFHM24-fisher-constant}.}

\author{Argyrios Deligkas}
\orcid{0000-0002-6513-6748}
\affiliation{%
    \institution{Royal Holloway}
    \city{Egham}
    \country{United Kingdom}}
\email{argyrios.deligkas@rhul.ac.uk}

\author{John Fearnley}
\orcid{0000-0003-0791-4342}
\affiliation{%
\institution{University of Liverpool}
\city{Liverpool}
\country{United Kingdom}}
\email{john.fearnley@liverpool.ac.uk}

\author{Alexandros Hollender}
\orcid{0000-0001-5255-9349}
\affiliation{%
\institution{University of Oxford}
\city{Oxford}
\country{United Kingdom}}
\email{alexandros.hollender@cs.ox.ac.uk}

\author{Themistoklis Melissourgos}
\orcid{0000-0002-9867-6257}
\affiliation{%
\institution{University of Essex}
\city{Colchester}
\country{United Kingdom}}
\email{themistoklis.melissourgos@essex.ac.uk}

\renewcommand{\shortauthors}{A. Deligkas, J. Fearnley, A. Hollender, and T. Melissourgos}

\begin{abstract}
We study the problem of computing approximate market equilibria in Fisher
markets with separable piecewise-linear concave (SPLC) utility functions. In
this setting, the problem was only known to be \ppad/-complete for
inverse-polynomial approximations. We strengthen this result by showing
\ppad/-hardness for constant approximations. This means that the problem does
not admit a polynomial time approximation scheme (PTAS) unless $\ppad/ = \p/$.
In fact, we prove that computing any approximation better than $1/11$ is
\ppad/-complete. As a direct byproduct of our main result, we get the same inapproximability
bound for Arrow-Debreu exchange markets with SPLC utility functions.
\end{abstract}

\begin{document}

\maketitle

\pagestyle{plain}
\pagenumbering{arabic}

\section{Introduction}
\label{sec:intro}
Fisher markets~\cite{brainard2005compute-fisher} are one of the
foundational models that have shaped modern economics. In a Fisher market, every
buyer has a fixed budget that they spend on their most favorable bundle of
goods according to their utility function. The appealing fact for this type of
markets is that when the utilities of the buyers satisfy some {\em ``standard
sufficient conditions''}, then a {\em market equilibrium} is {\em guaranteed} to exist.
In a market equilibrium, also known as a \emph{competitive equilibrium}, the prices and the allocation of the goods to the
buyers are such that: (a) every
buyer is allocated goods that maximize their utility, and (b) the market {\em
clears}, i.e., the supply of each good is exactly equal to the demand of that
good.

When the utility functions of the buyers are linear, a market equilibrium can always be computed in polynomial time~\cite{Devanur2008-market_equilibrium,orlin2010improved,vegh2012strongly}. However, the problem becomes intractable as soon as one considers a very slight generalization of linear utilities: {\em additive separable piecewise linear concave} (SPLC) utilities. A buyer with an
SPLC utility function has a piecewise linear concave utility for each good, and
their utility for a bundle of goods is simply the sum of their utilities for each
of the individual goods.

Computing a market equilibrium in Fisher markets with SPLC utilities was shown
to be a \ppad/-complete problem~\cite{vazirani2011market-plc-ppad,chen2009spending}, which
means that the problem is unlikely to admit a polynomial-time algorithm unless
\ppad/ = \p/. This motivates the study of \emph{approximate equilibria} in
which the condition that the market clears is replaced with an {\em approximate
clearing} constraint. In an \eps-approximate market equilibrium, we first
normalize the market so that there is exactly one unit of each good, and then
we seek a price vector such that every good \eps-clears, meaning that the
discrepancy between between the supply and the demand for any good is at most
$\eps$.

\citet{vazirani2011market-plc-ppad} actually showed that it is
\ppad/-complete to find an \eps-approximate market equilibrium when \eps is
inversely polynomial in the size of the market, i.e., the number of buyers and
the number of goods. 
This ruled out fully polynomial-time approximation schemes (FPTAS) for the problem unless $\ppad/ = \p/$. 
However, the existence of a polynomial-time approximation scheme (PTAS) was not
ruled out by this result, and a PTAS may well be good enough to
clear a market for most practical purposes.

\paragraph{\bf Our contribution.}

This paper provides the first \ppad/-completeness for \eps-approximate market
equilibria in Fisher markets for a {\em constant} \eps. 
Hence, a PTAS for Fisher markets cannot exist unless $\ppad/ = \p/$. 

\begin{theorem}\label{thm:main}
It is \ppad/-complete to compute an $\eps$-approximate market equilibrium in Fisher markets with SPLC utilities
for any constant $\eps < 1/11$.
\end{theorem}

We note that our hardness threshold, $1/11$, is relatively large. If we recall
that all goods are normalized to have one unit available, then this result
states that it is \ppad/-complete to find any approximate market equilibrium in
which all of the goods clear to within $9\%$ of their total supply. So this
essentially rules out the existence of polynomial-time algorithms that can
clear a market in practice.

\paragraph{\bf Technical Overview.}
We prove our result via a reduction from the \pcircuit problem that was
recently introduced by \citet{DFHM24}. In this problem we are given a circuit
with \NOT, \NAND, and \PURIFY gates, and we are required to find a satisfying
assignment using the values $\{0, 1, \garbo\}$. The \NOT gate behaves
as usual on the values in $\{0, 1\}$, and has no constraints when the input is
\garbo. The \NAND gate is required to output 0 if both inputs are 1, and is
required to output 1 if at least one input is 0, and otherwise has no
constraints. The \PURIFY gate has one input and two outputs. If the input is in
$\{0, 1\}$, then that value should be copied to both outputs, while if the
input is $\garbo$, then at least one output should be in $\{0, 1\}$. 

The use of the \pcircuit problem as a starting point for our reduction enables
us to give a construction that is arguably considerably simpler than previous
hardness results for Fisher markets. Prior work has used constructions in
which one first sets up a price-regulating market, which essentially gives a
set of goods that are able to encode values within a certain range, and then
adds extra buyers to simulate a \ppad/-hard problem, e.g., finding a Nash
equilibrium of a bimatrix game. 

In contrast, the market that we construct essentially directly implements the
gates of the \pcircuit instance. For each variable in the \pcircuit we
introduce a good, and the price of this good encodes the value of the variable.
Specifically we define a \emph{low price} $L$ and a \emph{high price} $H$ with
$L < H$, and then a price below $L$ encodes a $0$, a price above $H$ encodes a
1, while a price between the two encodes $\garbo$. 


Then we show that the \NOT and \NAND gates can be implemented by relatively
simple gadgets that each contain two buyers. The \PURIFY gate is implemented using two
chains of \NOT gates, where we carefully tweak the parameters of each of the
\NOT gates to ensure that the constraints of the \PURIFY gate are implemented.

\paragraph{\bf Arrow-Debreu exchange markets.}
Finally, we also show that our hardness result for Fisher markets implies a new
hardness result for  Arrow-Debreu exchange markets with SPLC utilities. Specifically, we use a
well-known reduction from Fisher markets to exchange markets that preserves
hardness for $\eps$-equilibria, and we use this to show that computing an
$\eps$-equilibrium in an exchange market is also \ppad/-hard for all $\eps <
1/11$. 

While hardness for Arrow-Debreu exchange markets was already known for constant
$\eps$ \cite{Rubinstein18-Nash-inapproximability}, prior work had only shown
this for some extremely small unknown constant, and so that work did not rule
out the existence of a polynomial-time algorithm that could clear a market in
practice. Our hardness result effectively rules this out unless \ppad/=\p/.

\subsection{Related work}
The problem of computing market equilibria, for both exact and approximate
equilibria, has
received significant attention over the years. For Fisher markets,
\citet{vazirani2011market-plc-ppad} and \citet{chen2009spending} have
established \ppad/-hardness for SPLC utilities albeit for a sub-constant \eps.
On the other hand, polynomial-time algorithms were derived for the cases where
the utility functions of the buyers are
linear~\cite{Devanur2008-market_equilibrium,orlin2010improved,vegh2012strongly},
homogeneous~\cite{eisenberg1961aggregation}, or weak gross
substitutes~\cite{codenotti2005polynomial}.

Furthermore, when the number of goods is constant \citet{kakade2004graphical}
gave a PTAS while~\citet{devanur2008market} gave a polynomial-time algorithm
for exact equilibria. In fact, the algorithm of~\cite{devanur2008market} also works
when the number of buyers is constant in the SPLC utility setting.
For non-separable PLC utilities~\citet{garg2022approximating} derived a fixed
parameter approximation scheme that has the number of buyers as a parameter.

Matching markets is another important subclass of general Fisher markets.
\citet{alaei2017computing} designed a polynomial-time algorithm for markets with a 
constant number of goods or buyers, while~\citet{vazirani2020computational}
derived a polynomial algorithm when the buyers have dichotomous utilities. For
one-sided matching markets, the most famous problem is the Hylland-Zeckhauser
market, for which the existence of an equilibrium was initially established
in~\cite{hylland1979efficient} and was recently simplified by
~\citet{braverman2021optimization}. Even more recently,
\citet{chen2022computational-Hylland-Zeckhauser-PPAD-h} have established
\ppad/-completeness for the problem.

There has also been interest in Fisher markets with
additional constraints. \citet{birnbaum2010new}, \citet{devanur2004spending},
and~\citet{vazirani2010spending} considered the case where the utilities of the
buyers depend on prices of goods through spending constraints.
\citet{jalota2023fisher} considered additional linear constraints that include
matching markets, and they gave a tâtonnement process which was
found to converge to a market equilibrium in experiments.

The Arrow-Debreu {\em exchange} market model~\cite{ArrowD54} is another
foundational class of markets. In this setting, the goods are brought to the
market by the buyers, who then spend the revenue they get from selling their initial
endowments. \ppad/-hardness for Arrow-Debreu market equilibria has been
established for several different
settings~\cite{ChenDDT09-Arrow-Debreu,ChenPY17-non-monotone-markets,codenotti2006leontief,garg2023auction}
and in fact \citet{Rubinstein18-Nash-inapproximability} showed that it is
\ppad/-hard to compute an \eps-equilibrium in Arrow-Debreu markets for a
small unknown constant \eps. On the positive side, there are
several polynomial-time algorithms for linear
utilities~\cite{duan2015combinatorial,duan2016improved,garg2019strongly,jain2007polynomial,ye2008path},
or for buyers with weak gross substitutes utilities~\cite{bei2019ascending,codenotti2005polynomial,garg2023auction}. The various types of market models have also been studied for the setting where the items are chores that provide disutility to the agents~\cite{BranzeiS24-chores,BoodaghiansCM2022-approx-CE-chores,ChaudhuryGMM22-existence-chores,ChaudhuryGMM22-CE-chores}.

The \pcircuit problem was recently introduced in~\cite{DFHM24}, where it was used
to prove strong, improved, \ppad/-hardness results for a variety of problems
related mainly to approximate Nash equilibria. Since then it was further used
in~\cite{deligkas2023tight} to prove tight \ppad/-hardness for approximate Nash
equilibria in graphical games, and in~\cite{ioannidis2023clearing} to prove
improved stronger \ppad/-hardness in the problem of clearing financial
networks. To the best of knowledge, this is the first time that \pcircuit has
been used to prove hardness for market equilibria.

\medskip

\noindent
\textbf{Subsequent work.}
After the publication of this work, a subsequent paper has shown that the
problem remains \ppad/-hard for any $\eps < 1/9$~\cite{DFHM26}. In
addition, that paper also studies the setting in which the optimality of the
bundles received by the buyers is allowed to be relaxed and establishes
conditional lower bounds for this version of the problem.

\section{Preliminaries}

\subsection{Fisher Markets}

\paragraph{\bf Fisher markets.} A Fisher market is given by a tuple $(G,B,(e_i)_{i\in B},(u_i)_{i \in B})$, where:
\begin{itemize}
\item $G$ is a set of (divisible) goods. Without loss of generality, we assume that there is one unit of each good available.\footnote{This can be achieved by a simple normalization, and it simplifies the expression for the clearing constraint below.}
\item $B$ is a set of buyers.
\item For every $i \in B$, $e_i > 0$ is the budget of buyer $i$.
\item For every $i \in B$, $u_i: \mathbb{R}_{\geq 0}^{|G|} \to \mathbb{R}_{\geq 0}$ is the utility function of buyer $i$. For any allocation $x_i \in \mathbb{R}_{\geq 0}^{|G|}$ of goods to buyer $i$ (where $x_{i,j} \geq 0$ denotes the amount of good $j$ allocated to buyer $i$), $u_i(x_i)$ denotes the utility derived by the buyer. We assume that the utility functions are separable piecewise-linear concave (SPLC), meaning that $u_i(x_i)$ can be written as $\sum_{j\in G} u_{i,j}(x_{i,j})$, where each $u_{i,j}: \mathbb{R}_{\geq 0} \to \mathbb{R}_{\geq 0}$ satisfies
\begin{enumerate}
    \item $u_{i,j}(0) = 0$,
    \item $u_{i,j}$ is continuous and piecewise-linear,
    \item $u_{i,j}$ is concave but non-decreasing.
\end{enumerate}
\end{itemize}

\paragraph{\bf Optimal bundles.}
Given a price vector $p \in \mathbb{R}_{\geq 0}^{|G|}$, where $p_j$ denotes the price of good $j$, the set of optimal bundles for buyer $i$, denoted $\opt_i(p) \subseteq \mathbb{R}_{\geq 0}^{|G|}$, is the set of optimal solutions of the following optimization problem:
\begin{equation}\label{eq:opt-bundle}
\begin{split}
\max \quad &u_i(x_i) \\
\text{ s.t.} \quad  
& \sum_{j \in G} p_j x_{i,j} \leq e_i \\
& x_{i,j} \geq 0 \quad \forall j \in G.
\end{split}
\end{equation}
Note that it is possible that $\opt_i(p) = \emptyset$, if some good has price $0$, and the agent is never satiated with this good.

\paragraph{\bf Competitive equilibrium.}
For any $\eps \geq 0$, an $\eps$-approximate market equilibrium is a price vector $p$ and an allocation vector $x = (x_i)_{i \in B}$ satisfying the following conditions:
\begin{enumerate}
    \item For each buyer $i$, $x_i$ is an optimal bundle at prices $p$, i.e., $x_i \in \opt_i(p)$.
    \item For each good $j$, the market clears approximately up to $\eps$ units of good, i.e.,
    $\left| \sum_{i \in B} x_{i,j} - 1 \right| \le \eps$.
\end{enumerate}
When $\eps = 0$, this corresponds to an exact market equilibrium.

\paragraph{\bf Existence of equilibria.}
The following condition is sufficient to guarantee the existence of a market equilibrium \cite{Maxfield1997,vazirani2011market-plc-ppad}:

\begin{description}
    \item[Sufficient Condition:] For every buyer $i \in B$, there exists a good $j \in G$ such that $u_{i,j}$ is a strictly increasing function (i.e., buyer $i$ is never satiated with good $j$).
\end{description}

\paragraph{\bf Computational problem.}
Let $\eps \geq 0$. The computational problem of computing an $\eps$-approximate market equilibrium is defined as follows: 
\begin{mdframed}[backgroundcolor=white!90!gray,
      leftmargin=\dimexpr\leftmargin-20pt\relax,
      innerleftmargin=4pt,
      innertopmargin=1pt,
      skipabove=5pt,skipbelow=5pt]
\begin{description}
    \item[Input:] A Fisher market $(G,B,(e_i)_{i\in B},(u_i)_{i \in B})$ with SPLC utilities satisfying the sufficient condition for the existence of equilibria. For each $i \in B$ and $j \in G$, $u_{i,j}$ is explicitly described in the input, i.e., for each linear affine piece we are given the positions and values at its endpoints.
    \item[Output:] An $\eps$-approximate market equilibrium $(p,x)$.
\end{description}
\end{mdframed}

Given $(p,x)$, the equilibrium conditions can be verified in polynomial time, because, for SPLC utilities, the optimization problem \eqref{eq:opt-bundle} defining $\opt_i(p)$ can be solved in polynomial time using a simple greedy approach;\footnote{In fact, given only prices $p$, it is possible to check in polynomial time whether there exists an allocation $x$ such that $(p,x)$ is an $\eps$-approximate market equilibrium \cite{vazirani2011market-plc-ppad}. So we could have equivalently defined the computational problem to only seek equilibrium prices $p$.} see, e.g., \cite{Garg2015-pivot}. Together with the existence of solutions guaranteed by the sufficient condition, this puts the problem in the complexity class \tfnp/ of total \np/ search problems. Prior work~\cite{vazirani2011market-plc-ppad} has shown that the problem lies in the subclass \ppad/ of \tfnp/, even for $\eps = 0$. In particular, exact rational solutions are guaranteed to exist. The problem is known to be \ppad/-complete for $\eps = 0$, and also when $\eps$ is part of the input and inverse-polynomial with respect to the description of the market~\cite{chen2009spending,vazirani2011market-plc-ppad}. No hardness result is known for any constant $\eps > 0$.

\subsection{The \pcircuit Problem}

\paragraph{\bf The \pcircuit problem.}

\begin{figure*}[t!]
    \begin{center}
	\begin{minipage}{0.27\textwidth}
		\begin{center}
			\begin{tabular}{c||c}
				$u$ & $v$ \\ \hline
				0 & 1 \\
				1 & 0 \\
				$\garbo$ & $\{0, 1, \garbo\}$\\
				\multicolumn{2}{c}{}
			\end{tabular}
			\caption*{\NOT gate}
		\end{center}
	\end{minipage}
    \begin{minipage}{0.35\textwidth}
        \begin{center}
            \begin{tabular}{c|c||c}
                $u$ & $v$ & $w$ \\ \hline
                1 & 1 & 0 \\
                0 & $\{0,1,\garbo\}$ & 1 \\
                $\{0,1,\garbo\}$ & 0 & 1 \\
                \multicolumn{2}{c||}{Else} & $\{0,1,\garbo\}$
            \end{tabular}
			\caption*{\NAND gate}
        \end{center}
    \end{minipage}
	\begin{minipage}{0.35\textwidth}
		\begin{center}
			\begin{tabular}{c||c|c}
				$u$ & \phantom{xx}$v$\phantom{xx}  & $w$ \\ \hline
				$0$ & $0$ & $0$ \\
				$1$ & $1$ & $1$ \\
				\multirow{2}{*}{$\garbo$} & \multicolumn{2}{c}{At least one} \\
				& \multicolumn{2}{c}{output in  $\{0, 1\}$}
			\end{tabular}
			\caption*{\PURE gate}
		\end{center}
	\end{minipage}
    \end{center}
	\caption{The truth tables of the three gates of \pcircuit.}
	\label{fig:gates}
\end{figure*}

We will show our hardness result by reducing from the \pcircuit problem, which
is known to be \ppad/-complete~\cite{DFHM24}. 
An instance of the \pcircuit problem is given by a node set $V=[n]$ and a set
$C$ of gate-constraints (or just \emph{gates}). Each gate $g \in C$ is of the
form $g = (T,u,v,w)$ where $u,v,w \in V$ are distinct nodes, and $T \in \{\NOT, \NAND, \PURE\}$ is the type of the gate, with the following interpretation.
\begin{itemize}
	\item If $T=\NOT$, then $u$ is the input of the gate, and $v$ is its output. ($w$ is unused)
	\item If $T=\NAND$, then $u$ and $v$ are the inputs of the gate, and $w$ is its output.
	\item If $T=\PURE$, then $u$ is the input of the gate, and $v$ and $w$ are its outputs.
\end{itemize}
We require that each node is the output of exactly one gate.

A solution to instance $(V,C)$ is an assignment $\valonly: V \to \{0,1,\garbo\}$ that satisfies all the gates (see Fig.~\ref{fig:gates}), i.e., for each gate $g=(T,u,v,w) \in C$ we have the following. 
\begin{itemize}
	\item If $T=\NOT$ in $g=(T,u,v)$, then $\valonly$ satisfies
	\begin{align*}
		&\val{u} = 0 \implies \val{v} = 1\\
		&\val{u} = 1 \implies \val{v} = 0.
	\end{align*}

    \item If $T=\NAND$ in $g = (T,u,v,w)$, then $\valonly$ satisfies
    \begin{align*}
        \val{u} = \val{v} = 1 \implies \val{w} = 0\\
        (\val{u} = 0) \lor (\val{v} = 0) \implies \val{w} = 1
    \end{align*}

	\item If $T=\PURE$, then $\valonly$ satisfies
	\begin{align*}
		& \{\val{v}, \val{w}\} \cap \{0,1\} \neq \emptyset\\
		& \val{u} \in \{0,1\} \implies \val{v} = \val{w} = \val{u}.
	\end{align*}
\end{itemize}
 
The structure of a \pcircuit instance is captured by its \emph{interaction graph}. This graph is constructed on the vertex set $V = [n]$ by adding a directed edge from node $u$ to node $v$ whenever $v$ is the output of a gate with input $u$. The total degree of a node is the sum of its in- and out-degrees.

\begin{theorem}[\cite{DFHM24}]
	\label{thm:pancircuit}
	\pcircuit is \ppad/-complete, even when every node of the interaction graph has in- and out-degree at most 2 and total degree at most 3.
\end{theorem}

\section{Construction of the Market}

Given an instance $(V, C)$ of the \pcircuit problem, we will construct, in
polynomial time, a Fisher market that simulates that instance. In this section, we
describe how the Fisher market will be constructed. 

\paragraph{\bf The reference good.}

To implement our construction, we need a specific good that, in every
equilibrium, has a price that is close to $1$. We call this good the
\emph{reference good}, and denote it as \refg.

To ensure that the price of the reference good is close to $1$, we use a
\emph{reference buyer} called \refb who has a budget of $e_{\refb} = 1$, and the following
utility function.
\begin{equation*}
u_{\refb, j}(x) = \begin{cases}
x & \text{if $j = \refg$,} \\
0 & \text{otherwise.}
\end{cases}
\end{equation*}
In other words, the reference buyer desires only the reference good, and will
therefore spend all of their money on it. 

We will ensure that the total amount of demand on \refg from all other buyers
in the construction will be significantly smaller than $1$. This will ensure
that in any approximate equilibrium, the price of \refg will be close to $1$. 

\paragraph{\bf Variable encodings.}

For each variable in the \pcircuit instance, we introduce a good that will
encode that variable. The value assigned to each variable will be determined by
the price of the corresponding good. 

Let $s, a \in \reals$ be two positive constants that
will be fixed later. Given a price vector $p$, we define a \emph{high price}
$H = s \cdot p_{\refg}$, and a \emph{low price} $L = s \cdot H / a = s^2 \cdot
p_{\refg} / a$. The idea is that we will fix $s < 1$, while $a$ will be chosen
to be very large. Thus, the high price is a specified fraction of the reference
price, while the low price is very close to zero. 

Given a price vector $p$ for the Fisher market, we extract an assignment to the variables of the \pcircuit instance in the following way.
\begin{itemize}
\item If $p_v \ge H$ then $v = 1$.
\item If $p_v \le L$ then $v = 0$.
\item Otherwise $v = \garbo$. 
\end{itemize}

Note that the values of $H$ and $L$ depend on the price of the reference good.
Although we know that the price of the reference good will be close to $1$, it
will not be exactly $1$, and thus $H$ and $L$ will vary according to the
particular price that is chosen for the reference good. We will use \hlow and
\hhigh to denote a lower and upper bound for $H$, and we will give exact values
for these bounds later. We likewise use \llow and \lhigh to give bounds for
$L$.

\paragraph{\bf Auxiliary buyers.}

The construction will use many \emph{auxiliary buyers}, whose purpose is to buy
a pre-specified amount of a particular good. 
Given a good $j \in G$ and an amount $r \in [0, 1]$ we define the buyer
$b = \aux(j, r)$ in the following way.
\begin{itemize}
\item The buyer's budget is $e_b = r \cdot \hhigh$.
\item The buyer's utility for good $j$ is defined to be
\begin{equation*}
u_{b, j}(x) = \begin{cases}
2 \cdot s \cdot x & \text{if $x \le r$,} \\
2 \cdot s \cdot r & \text{otherwise.}
\end{cases}
\end{equation*}
\item The buyer's utility for \refg is $u_{b, \refg}(x) = x$.
\item The buyer's utility for all other goods is zero.
\end{itemize}

We will ensure that the price of good $j$ is no larger than \hhigh, which means
that the auxiliary buyer always has enough money to buy $r$ units of 
good~$j$. The utilities have been chosen to ensure that, if the good's price
is no larger than \hhigh, the auxiliary buyer always strictly prefers to buy $r$ units
of good $j$ before buying the reference good. Moreover, once $r$ units of good
$j$ have
been bought by the auxiliary buyer, their marginal utility for good $j$ becomes
$0$, whereas they always have positive marginal utility towards the reference
good.
Thus, as we will later
show formally, in any approximate equilibrium the auxiliary buyer will buy
exactly $r$ units of good $j$, and spend the rest of their money on the
reference good. 

\paragraph{\bf The interface between variables.}

As mentioned previously, for each variable in the \pcircuit instance, there is
a corresponding good whose price encodes the value of that variable. To
simulate the gates of the \pcircuit instance, we will use buyers that buy a
certain proportion of the variable-encoding goods. 

To ensure that there is a consistent interface between the gates, we introduce
a parameter $t \in (0, 0.5)$ with the intention that, if a gate $g$ uses variable
$v$ as an input variable, then the buyers that implement $g$ will buy $t$ units
of the good that represents $v$. 

From Theorem~\ref{thm:pancircuit} we have that the \pcircuit instance has
out-degree at most two, so for each variable-encoding good, we expect at most $2t$
units of the good to be bought by the gates that take this variable as an
input. Note however that some goods may not be used as an input to exactly two
gates. To address this, for each variable-encoding good $j$ that is used as an input to 
only one gate, we introduce a buyer $\aux(j, t)$, who will buy $t$ units of that
good to top-up the amount that is bought. 
In the case where $j$ is used as an input to no other gate, we introduce two such buyers instead.

In a perfect world, this would ensure that exactly $2t$ units of each
variable-encoding good are bought by the gates for which that good is an input.
Unfortunately, as we will later show formally, the gates may actually buy
slightly less than $t$ units each, and this will be dealt with in our proofs.

\paragraph{\bf \NOT gates.}

We begin by building a \NOT gate. The \NOT gate is particularly important
because we will also use it to implement the \PURIFY gate. 
For this reason, we implement a parameterized gadget $\NOT(\ing, \outg, r)$,
where $\ing$ is the input good to the \NOT gate, \outg is the output good, and
$r \in [0,1]$ is a constant. When the gadget is used to implement a \NOT gate
in the \pcircuit instance, we will set $r = \rnot$, where $\rnot \in [0,1]$ is
a fixed parameter.

The gadget will have a buyer named the inverter,
and also an auxiliary buyer. The interaction between the buyers and the goods
is shown below, where an arrow indicates that the buyer has a non-zero utility
function for that good.

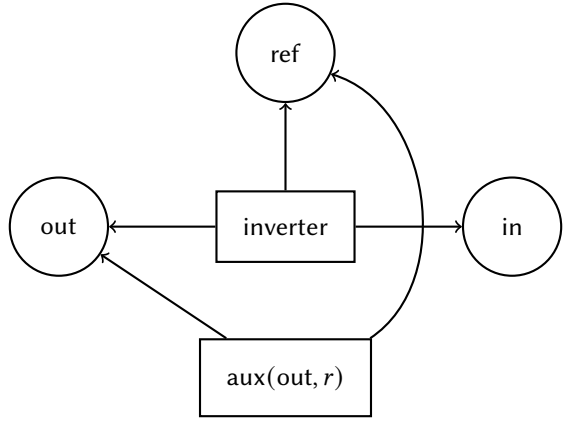
\begin{center}
\begin{tikzpicture}

\node[circle,draw,thick,minimum size=1.3cm] (gout) at (0, 0) {\outg};
\node[circle,draw,thick,minimum size=1.3cm] (gin) at (6, 0) {\ing};
\node[circle,draw,thick,minimum size=1.3cm] (ref) at (3, 2.3) {$\refg$};

\node[draw, inner sep=0.35cm,thick] (buyer) at (3, 0) {\inverter};
\node[draw, inner sep=0.35cm,thick] (extra) at (3, -2) {$\aux(\outg, r)$};

\path[draw,thick,->] 
    (buyer) edge (gout)
    (buyer) edge (gin)
    (buyer) edge (ref)
    (extra) edge (gout)
    (extra) edge [bend right=65] (ref)
    ;
\end{tikzpicture}
\end{center}

The inverter is specified as follows.
\begin{itemize}
\item The budget of the inverter is $e_{\inverter} = t \cdot \hlow$.
\item The inverter's utility for \ing is
\begin{equation*}
u_{\inverter,\ing}(x) = \begin{cases}
a \cdot x & \text{if $x \le t$,} \\
a \cdot t & \text{otherwise.}
\end{cases}
\end{equation*}
\item The inverter's utility for \outg is $u_{\inverter,\outg}(x) = s \cdot x$.
\item The inverter's utility for \refg is $u_{\inverter,\refg}(x) = x$.
\item The inverter's utility for all other goods is zero.
\end{itemize}

Recall that we intend to set $s < 1$ and $a$ to be very large. Thus, the
inverter buyer is heavily incentivized to buy $t$ units of the input good
before buying anything else. However, once $t$ units of the input good have
been bought, the inverter no longer has any interest in buying more units.
We call $t$ the \emph{anti-endowment} for the inverter: that buyer
must\footnote{Actually, if the price of the output good is very close to zero, then the
inverter may prefer to buy the output good before the input good. We will deal
with this case separately in our proofs.}
buy a specific amount of the good before being able to spend money elsewhere.

The high-level idea is that if the input variable is $1$, then $p_\ing
\ge H \ge \hlow$. Since we have fixed the budget of the inverter to be
$e_{\inverter} = t \cdot \hlow$, this means that the inverter will
spend all of their budget buying their anti-endowment, and will have no money
left to spend on the output good. Thus, the demand on the output good will be
low, and so its price cannot be high.

On the other hand, if the input variable is $0$, then $p_\ing \le L = s
\cdot H/a$. Since we will set $a$ to be very large, this means that the
inverter spends almost no money buying their anti-endowment, and so has
essentially their entire budget left over. In this scenario, the inverter will
have a large amount of money to spend on the output good, which will cause it to have a
high price.

The utility of the inverter towards the reference good serves to ensure that
the output good's price cannot rise above $H$. Since $H = s \cdot p_\refg$, if
$p_\outg > H$ then the marginal utility of buying the output good would be
strictly less than $s / (s \cdot p_\refg) = 1/p_\refg$, whereas the marginal utility of buying the reference good is $1/p_\refg$, and so in this scenario the
inverter would spend all of their remaining money on the reference good, and no
money on the output good. We will ensure that the output good fails to clear in
this scenario, which will ensure that the price of the output good is capped at
$H$ in any approximate equilibrium.

The auxiliary buyer allows us to change how much of the output good can be
bought by the inverter, since we can adjust the $r$ parameter to
change how much of the output is taken by the auxiliary buyer. The high-level
idea here is that by changing $r$, we can change how the output good's price
changes relative to the inverter's remaining budget after buying the
anti-endowment. This will be used critically when we use \NOT gates to
implement the \PURIFY gate.

\paragraph{\bf \NAND gates.}

For each \NAND gate with inputs  \inoneg and \intwog, and output \outg, we use
the following construction.

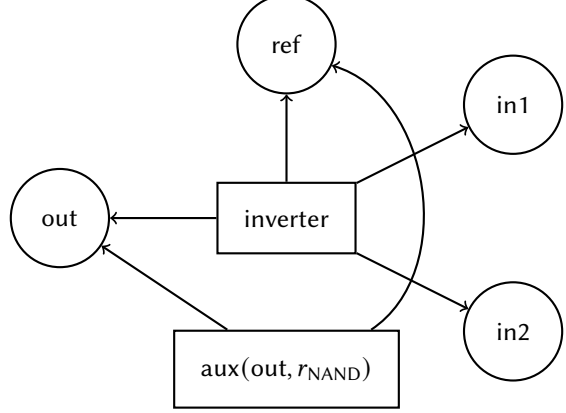
\begin{center}
\begin{tikzpicture}

\node[circle,draw,thick,minimum size=1.3cm] (gout) at (0, 0) {\outg};
\node[circle,draw,thick,minimum size=1.3cm] (gin1) at (6, 1.5) {\inoneg};
\node[circle,draw,thick,minimum size=1.3cm] (gin2) at (6, -1.5) {\intwog};
\node[circle,draw,thick,minimum size=1.3cm] (ref) at (3, 2.3) {$\refg$};

\node[draw, inner sep=0.35cm,thick] (buyer) at (3, 0) {\inverter};
\node[draw, inner sep=0.35cm,thick] (extra) at (3, -2) {$\aux(\outg, \rnand)$};

\path[draw,thick,->] 
    (buyer) edge (gout)
    (buyer) edge (gin1)
    (buyer) edge (gin2)
    (buyer) edge (ref)
    (extra) edge (gout)
    (extra) edge [bend right=65] (ref)
    ;
\end{tikzpicture}
\end{center}

The construction uses two buyers. The auxiliary buyer has a parameter \rnand
that will be fixed later. 
The \inverter buyer is slightly different than that of the \NOT gate and is specified as follows.

\begin{itemize}
\item The budget of the inverter is $e_{\inverter} = 2 t \cdot \hlow$.
\item The inverter's utility for good $j \in \{\inoneg, \intwog\}$ is
\begin{equation*}
u_{\inverter,j}(x) = \begin{cases}
a \cdot x & \text{if $x \le t$,} \\
a \cdot t & \text{otherwise.}
\end{cases}
\end{equation*}
\item The inverter's utility for \outg is $u_{\inverter,\outg}(x) = s \cdot x$.
\item The inverter's utility for \refg is $u_{\inverter,\refg}(x) = x$.
\item The inverter's utility for all other goods is zero.
\end{itemize}

This is a straightforward generalization of the \NOT gate to two inputs. The
inverter now has a budget of $2 t \cdot \hlow$ and an anti-endowment of $t$
units for both of the input goods. So if both inputs have a high price, the
inverter will spend all their money on the input goods, and so they will not be
able to increase the price of the output good. If either of the two inputs has
a low price, then the inverter will have money left over, and will be able to
push the price of the output good higher.

\paragraph{\bf \PURIFY gates.}

For a \PURIFY gate with input \ing and outputs \outoneg and \outtwog, we use
the following construction.

\begin{center}
\resizebox{\textwidth}{!}{
\begin{tikzpicture}

\node[circle,draw,thick,minimum size=1.3cm] (gin) at (6, 0) {\ing};

\node[draw, inner sep=0.35cm,thick] (buyer11) at (3, 1) {$\NOT(r^1_1)$};
\node[draw, inner sep=0.35cm,thick] (buyer21) at (3, -1) {$\NOT(r^2_1)$};

\node[circle,draw,thick,minimum size=1.3cm] (g11) at (0, 1) {$g^1_1$};
\node[circle,draw,thick,minimum size=1.3cm] (g21) at (0, -1) {$g^2_1$};

\node[draw, inner sep=0.35cm,thick] (buyer12) at (-3, 1) {$\NOT(r^1_2)$};
\node[draw, inner sep=0.35cm,thick] (buyer22) at (-3, -1) {$\NOT(r^2_2)$};

\node[circle,draw,thick,minimum size=1.3cm] (g12) at (-6, 1) {$g^1_2$};
\node[circle,draw,thick,minimum size=1.3cm] (g22) at (-6, -1) {$g^2_2$};

\node[circle,thick,minimum size=1.3cm] (dots1) at (-8, 1) {$\cdots$};
\node[circle,thick,minimum size=1.3cm] (dots2) at (-8, -1) {$\cdots$};

\node[circle,draw,thick,minimum size=1.3cm] (out1) at (-10, 1) {$\outoneg$};
\node[circle,draw,thick,minimum size=1.3cm] (out2) at (-10, -1) {$\outtwog$};

\path[draw,thick,->] 
    (buyer11) edge (gin)
    (buyer21) edge (gin)

    (buyer11) edge (g11)
    (buyer21) edge (g21)

    (buyer12) edge (g11)
    (buyer22) edge (g21)

    (buyer12) edge (g12)
    (buyer22) edge (g22)

    (dots1) edge (g12)
    (dots2) edge (g22)

    (dots1) edge (out1)
    (dots2) edge (out2)
    ;
\end{tikzpicture}
}
\end{center}

We use two chains of $d$ \NOT gates to compute the two outputs,
where $d$ is a parameter. To do this, we introduce intermediate goods $\{
g^1_j, g^2_j \; : \; 1 \le j \leq d-1 \ \}$. To simplify the definition, we use
$g^1_0 = g^2_0 = \ing$, and we use $g^1_d = \outoneg$ and $g^2_d = \outtwog$.
Then for each $i \in \{1, 2\}$, and for each $j \in \{1, 2, \dots, d\}$ we
include a gadget $\NOT(g^i_{j-1}, \; g^i_j, \; r^i_j)$ where $r^i_j \in \reals$
is a parameter that we will fix later. 

We treat each of the \NOT gates and each of the intermediate goods as full
gates and variables in our instance. This means that each intermediate
good $j$ will also have an auxiliary buyer $\aux(j, t)$ that buys $t$ units of
the good, to compensate for the fact that the good is only used as an input to
one other gate.

The idea is that we will set $d$ to be some large even number. Therefore each
chain of \NOT gates will have even length. So if the input variable is a $1$ or
$0$, both chains of \NOT gates will output that value, as required by the
\PURIFY constraints. If the input good has a price that is strictly between $L$
and $H$, meaning that it encodes a $\garbo$ value, then one of the two chains
is required to output a value in $\{0, 1\}$. We will ensure this by carefully
selecting values for the parameters $r^i_j$. 

More specifically, the parameters are chosen so that there exists a price $p_{\ing}^*$
such that, if the price of the input good is at least $p_{\ing}^*$, then the prices of
the goods in the top chain will increase and we will have that $p_{\outg1} = H$.
On the other hand, if $p_\ing$ is less than $p_{\ing}^*$, then the prices in the bottom
chain will decrease and we will have that $p_{\outg2} \le L$. This
implies that at least one of the two chains will always output a pure value no
matter what the input price is, as required by the \PURIFY gate.

\paragraph{\bf Circuit copies.}

So far we have described a full reduction from \pcircuit to a Fisher market,
and while this construction is strong enough to give hardness for a constant $\eps$,
the bound that we would obtain would be much smaller than $1/11$. The main reason for this is
the uncertainty in the price of the reference good. As we mentioned earlier,
this price should be close to $1$, but as $\eps$ increases, our bounds on it get
weaker. 

To address this, we will introduce $k$ copies of the circuit, where each copy
is only required to work when the reference price is within a particular range.
Specifically, letting $\hmin := s/2$ and $\hmax := 2s$, we divide the region $[\hmin,\hmax]$ into $k$ equally sized non-overlapping regions. Then for each such
region $[x, y]$, we build a copy of the circuit setting $\hlow = x$ and $\hhigh
= y$. As we will show, the price of the reference good in any \eps-equilibrium will be bounded by $\prefmin := 1/2$ and $\prefmax := 2$, and so we will indeed always have that $H \in [\hmin,\hmax]$.

When we are given an approximate equilibrium of the Fisher market, we first
find an interval $[x, y]$ that contains $H$, and then decode the
assignments to the \pcircuit instance from the copy that corresponds to that
interval, while ignoring all other copies. 

The key advantage of this is that each circuit copy can now assume $\hlow$ and
$\hhigh$ are very close together, which then increases the values of $\eps$
for which we can show hardness. 

\paragraph{\bf The sufficient condition.}

Finally, we will verify that the sufficient condition for the existence of an equilibrium holds for our
construction. Recall that this condition requires that 
for every buyer $i$, there exists a good $j$ such that $u_{i,j}$ is a strictly
increasing function. There are three types of buyers in our market: inverters, 
auxiliary buyers, and the reference buyer. All of these buyers have the same
utility for the reference good, and so for every buyer $i$ in the market we
have $u_{i, \refg}(x) = x$.
Therefore the sufficient condition is satisfied.

\section{Analysis}\label{sec:analysis}

Fix any $\eps < 1/11$. The construction described in the previous section uses several parameters. For
our proofs, we will fix these parameters to the following values. We first set $\delta := \frac{11}{4}\cdot(1/11 - \eps) > 0$.

\begin{center}
\begin{tabular}{l|r|l}
Parameter & Value & Description \\ \hline
$t$ & $4/11$ & The size of each inverter's anti-endowment \\
$d$ & $2 \ceil{\log_{2}(3/\delta)}$ & The length of the \NOT gate chains in a \PURIFY gate \\
$k$ & $110/\delta$ & The number of copies of the circuit \\
$s$ & $1/(20 k d |V|)$ & The inverter's marginal utility for the output good\\
$a$ & $\max\{2,4s/\delta\}$ & The inverter's marginal utility for the input good\\
$\rnot$ & 2/11 & The value of $r$ used in the \NOT gate \\
$\rnand$ & 2/11 & The value of $r$ used in the \NAND gate \\
$r^1_j$, $j$ is odd & 0 & The $r$ values used in the first \NOT chain in a
\PURIFY gate \\
$r^1_j$, $j$ is even & $2/11$ \\
$r^2_j$, $j$ is odd & $2/11$ & The $r$ values used in the second \NOT chain in a
\PURIFY gate \\ 
$r^2_j$, $j$ is even & 0 \\
\end{tabular}
\end{center}

For the rest of this section, we will consider an $\eps$-approximate market equilibrium $(p,x)$ of the market.

\paragraph{\bf General properties of the construction.}

Before we prove the correctness of each of the individual gates, we first prove
some general properties of the construction that will be useful later. We start
by considering the reference good. Recall that the reference good was intended
to have price close to 1. The following lemma shows this this is indeed the case.

\begin{lemma}
\label{lem:p_refg_bounds}
We have $p_\refg \in [\prefmin,\prefmax]$, where 
$\prefmin = 1/2$ and $\prefmax = 2$.
In particular, it follows that $H = s \cdot p_\refg \in [\hmin,\hmax]$,
where $\hmin = s/2$ and $\hmax = 2s$.
\end{lemma}

\begin{proof}
If $p_\refg = 0$, then, by construction, the reference buyer will demand an infinite amount of good $\refg$. In particular, $\opt_{\refb}(p) = \emptyset$, and we cannot be at an $\eps$-equilibrium. Thus, we must have $p_\refg > 0$. In that case, any optimal bundle for buyer $\refb$ will demand exactly $e_{\refb}/p_\refg = 1/p_\refg$ units of the reference good. If $p_\refg < \prefmin = 1/2$, then buyer $\refb$ demands $1/p_\refg > 2 > 1 + \eps$ units of good. Since there is only one unit of good available, this is a contradiction to the $\eps$-clearing condition. Thus, we must have $p_\refg \geq \prefmin$.

Let $E_{-\refb}$ denote the sum of budgets of all buyers, except the reference buyer $\refb$, i.e., $E_{-\refb} = \sum_{i \in B \setminus \{\refb\}} e_i$. By construction, the budget of any buyer $i \neq \refb$ satisfies $e_i \leq \hmax$. Indeed, the budget of any such buyer is either $t \cdot \hlow$, or $2t \cdot \hlow$, or $r \cdot \hhigh$ for some $r \in [0,1]$. Furthermore, we have $|B \setminus \{\refb\}| \leq 4 k d |V|$, since there are $k$ copies and in each copy there are at most $d |V|$ goods, and for each such good there are at most four buyers having non-zero utility for it. Thus, we can bound $E_{-\refb} \leq 4 k d |V| \hmax$.

We can now proceed to prove the upper bound on the price $p_\refg$. The total demand on the reference good is at most $(e_{\refb} + E_{-\refb})/p_\refg$. This corresponds to the case where all buyers spend all of their budget on the reference good. In order for the reference good to $\eps$-clear, we must thus have that $(e_{\refb} + E_{-\refb})/p_\refg \geq 1 - \eps$, i.e.,
$$p_\refg \leq \frac{e_{\refb} + E_{-\refb}}{1 - \eps} = \frac{1}{1-\eps} + \frac{E_{-\refb}}{1-\eps} \leq \frac{1}{1-\eps} + 5 k d |V| \hmax \leq \frac{1}{1-\eps} + 10 k d |V| s \leq 2 = \prefmax$$
where we used $\eps < 1/11$, $E_{-\refb} \leq 4 k d |V| \hmax$, $\hmax = 2s$, and $s \leq 1/(20 k d |V|)$.
\end{proof}

Recall that we have $k$ copies of the circuit, and that each of the circuits is
required to work only when $p_\refg$ is within the given range for that
circuit. 
From now on for the rest of this section we will focus only on the circuit copy
that was built for the particular value of $p_\refg$ in our equilibrium.
Observe that by construction we have $\hlow \leq H \leq \hhigh$ in this copy.
We can also apply \cref{lem:p_refg_bounds} to obtain the following bounds,
which will prove useful later. We have
\begin{equation}\label{eq:high-minus-low}
\hhigh - \hlow \leq 2s/k
\end{equation}
and thus
\begin{equation}\label{eq:low-over-high}
1 - 4/k \leq \hlow/\hhigh \leq 1
\end{equation}
where we used $\hhigh \geq \hmin = s/2$.

The next lemma gives, for each \NOT or \NAND gate, an upper bound on the amount
of the output good that can be allocated to buyers other than the inverter which uses
that good as an output.

\begin{lemma}
\label{lem:external_demand}
For any good $j \in G \setminus \{\refg\}$, let $i$ be the inverter that implements the gate using $j$ as an output variable. 
If the gate is a \NOT gate with parameter $r$, then the total amount of good
$j$ allocated to buyers other than $i$ is at most
$2t + r$, while if the gate is a \NAND gate, then the total amount 
$j$ allocated to buyers other than $i$ is at most
$2t + \rnand$.
\end{lemma}
\begin{proof}
The only buyers other than buyer $i$ who have non-zero utility toward good $j$ are:
\begin{itemize}
\item The auxiliary buyer from the \NOT or \NAND gate that outputs to good $j$.
A \NAND gate auxiliary will buy at most $\rnand$ units of good $j$, while an
auxiliary for a \NOT gate with parameter $r$ will buy at most $r$ units of good
$j$.
\item The inverters from gates that take good $j$ as an input, who each can buy at most $t$ units of good $j$.
\item Auxiliary buyers that buy at most $t$ units of good $j$ whenever that good is not used as an input to exactly two other gates.
\end{itemize}
These buyers cannot buy more than the amount specified above because once they
have been allocated their specified amount, their marginal utility for good $j$
becomes $0$, whereas by the sufficiency condition, all buyers have at least one good with positive marginal
utility for which they are never satiated. The same is true for all buyers that
have utility $0$ for good $j$.

Hence, for a \NOT gate with parameter $r$, the total amount demanded by buyers
other than buyer $i$ is at most $r + 2t$  while for a \NAND gate total amount
demanded by buyers other than buyer $i$ is at most $\rnand + 2t.$
\end{proof}

The next lemma states that for each \NOT or \NAND gate, the inverter must be
allocated strictly more than 0 units of the output good.

\begin{lemma}
\label{lem:non_zero_output}
For any good $j \in G \setminus \{\refg\}$, let $i$ be the inverter that
implements the \NOT or \NAND gate that has $j$ as its output variable. 
Then, buyer $i$ is allocated strictly more than $0$ units of good $j$.
\end{lemma}

\begin{proof}
From \cref{lem:external_demand} we have that buyers other than buyer $i$ can
buy at most 
\begin{align*}
& \max \left(\{\rnot, \rnand \} \cup \{ r^p_q \; : \; 1 \le p \le 2, 1 \le q \le d\}
\right) + 2t  \ge 2/11 + 8/11 = 10/11
\end{align*}
units of good $j$, where we have used the fact that $t = 4/11$, that $\rnot =
\rnand = 2/11$, and that each $r^p_q \le 2 \eps < 2/11$. 
Since $\eps < 1/11$ we have that $10/11 < 1 - \eps$, so
if $x_{i,j} = 0$, then good $j$ does not $\eps$-clear, and so we are not in an
$\eps$-equilibrium.
\end{proof}

Next we prove 
that price of every good other than the reference good
can have price at most $H$.

\begin{lemma}
\label{lem:p_j_bounds}
For any good $j \in G \setminus \{\refg\}$, $0 < p_j \leq H$.
\end{lemma}

\begin{proof}
We first prove that $p_j > 0$. For the sake of contradiction,
suppose that some good $j \in G \setminus \{\refg\}$ has $p_j = 0$. Recall that in the \pcircuit instance, each vertex is an output of exactly one gate. As a result, in our construction, for every good $j \in G \setminus \{\refg\}$ there is a buyer $i$, namely the inverter of the gate whose output is $j$, who is never satiated with $j$.
Since the price of good $j$ is $0$, this buyer will demand an infinite amount of good $j$, making $\opt_{i}(p) = \emptyset$. This implies that we cannot be in an \eps-equilibrium.

We now prove that $p_j \le H$.
So, for the sake of contradiction, suppose that $p_j > H$. By construction, as mentioned above, good $j$ is the output good of some gate, and in particular, an inverter uses this good as an output. That inverter's marginal utility for \refg is $1/p_{\refg}$, and for good $j$ the marginal utility is $s/p_j < s/H = 1/p_{\refg}$. Moreover, the inverter is never satiated with \refg, meaning that they
will demand zero units of good $j$. We can now use \cref{lem:non_zero_output} to obtain a contradiction.
\end{proof}

The next lemma states that the auxiliary players always do their jobs
correctly, meaning that they buy exactly as much of the target good as we have
specified.

\begin{lemma}
\label{lem:auxiliary}
For any auxiliary buyer $i \in B$ with target good $j \in G \setminus \{\refg\}$ and mandated amount $r \in [0,1]$, the buyer is allocated exactly $r$ units of good $j$, i.e., $x_{i,j} = r$.
\end{lemma}

\begin{proof}
We first argue that $x_{i, j} \le r$. This follows from the fact that, once
buyer $i$ is allocated $r$ units of good $j$, their marginal utility for good
$j$ is $0$, whereas their marginal utility for the reference good is
$1/p_\refg > 0$. Thus buyer $j$ strictly prefers the reference good to good $j$, and so cannot be allocated more than $r$ units of good $j$.  

Next we argue that $x_{i, j} \ge r$. If buyer $i$ has been allocated strictly
less than $r$ units of good $j$, then their marginal utility for good $j$ is
$2s/p_j$. From \cref{lem:p_j_bounds}, we have that $p_j \le H$, so their
marginal utility for good $j$ is at least $2s/H$. On the other hand, their
marginal utility for reference good is $1/p_\refg = s/H$, and their marginal
utility for all other goods is $0$. Thus buyer $i$ strictly prefers 
good $j$ to any other good. Moreover, since $e_i = r \cdot \hhigh \ge r \cdot H$ ,
buyer $i$ has enough money to buy $r$ units of good $j$. Thus any optimal
bundle must allocate at least $r$ units of good $j$ to buyer $i$.
\end{proof}

\subsection{Bounds on anti-endowment purchases}

Recall that our intention is that for each variable-encoding good, the gadgets
that take that good as an input should purchase $t$ units of the good as an
anti-endowment. As we
mentioned earlier, this is unfortunately not the case, and it is possible
that less that $t$ units are purchased by each gadget. In this section we 
formally prove bounds on how much each gadget purchases.

The following lemma shows that if the output good has price strictly greater
than $L$, then the inverter buyer must purchase their full anti-endowment
before buying any other good. In particular, since \cref{lem:non_zero_output}
requires the inverter to spend non-zero money on the output good, this means
that the inverter must buy $t$ units of all input goods.

\begin{lemma}
\label{lem:not_gate_big_l}
Let $j$ be an inverter in any \NOT or \NAND gate. 
If $p_\outg > L$ then
\begin{itemize}
\item if the gate is a \NOT gate then $j$ must buy $t$ units of
\ing before buying any other good; and
\item if the gate is a \NAND gate then $j$ must buy $t$ units of \inoneg and
$t$ units of \intwog before buying any other good.
\end{itemize}
\end{lemma}

\begin{proof}
If the gate is a \NOT gate, then the inverter's marginal utility for \ing is $a/p_\ing \geq a/H$, where the
second inequality comes from \cref{lem:p_j_bounds}. On the other hand, their
marginal utility for \outg is $s/p_\outg < s/L = a/H$, and their marginal
utility for $\refg$ is $1/p_\refg = s/H$. Since $s < a$, the inverter will buy
up to $t$ units of \ing before buying any other good.

For the case where the gate is a \NAND gate, we have that the inverter's
marginal utility for \inoneg is $a/p_\inoneg \geq a/H$, and the inverter's marginal
utility for \intwog is $a/p_\intwog \geq a/H$, where in both cases the second
inequality comes from \cref{lem:p_j_bounds}. We can now use the same argument
as above to conclude that the inverter must buy $t$ units of both \inoneg and
\intwog before buying any other good.
\end{proof}

The only remaining case is when $p_\outg \le L$. In this case the inverter
may actually purchase the output good before buying any of the input
goods. However, since $L$ is very close to zero, the maximum amount of money that the
inverter can spend on the output good without violating the $\eps$-clearing
constraint is also very small. This may cause the inverter to buy 
slightly less than
their
full anti-endowments. The following pair of lemmas give formal bounds on this
for \NOT gates and \NAND gates, respectively.

\begin{lemma}
\label{lem:not_gate_inverter}
For any \NOT gate, the inverter's allocation of the input good, \ing, satisfies $x_{\inverter,\ing} \in [t - 3/k ,t]$. 
\end{lemma}

\begin{proof}
We first consider the case where $p_\outg > L$. By \cref{lem:not_gate_big_l},
we have that the inverter must buy $t$ units of $\outg$ before buying any other
good. There are two sub-cases.
\begin{itemize}
\item If $e_\inverter \ge t \cdot p_\ing$ then the inverter has enough money to buy $t$ units of \ing, and does so.
\item If $e_\inverter < t \cdot p_\ing$ then the inverter does not have enough
money to buy $t$ units of \ing. This means that they spend all of their money
on \ing, and therefore they they will demand zero units of \outg, which
contradicts \cref{lem:non_zero_output}.
\end{itemize}

We now consider the case where $0 < p_\outg \leq L$. First note that the marginal utility of \outg is $s/p_\outg > s/L = a/H$, whereas the marginal utility of \refg is $1/p_\refg = s/H$. Since $a > s$, and since the inverter is never satiated by \outg, this means that the inverter cannot spend anything on \refg. Due to the $\eps$-clearing constraint, we know that the inverter can buy at most $1+\eps$ units of \outg, and so the inverter will spend at most $(1+\eps) \cdot L$ money on \outg. All remaining money must be spent on \ing, so the inverter will spend at least $e_\inverter - (1+\eps) \cdot L$ money on \ing. Therefore, since $p_\ing \leq H$ by \cref{lem:p_j_bounds}, the inverter will buy at least
\begin{align*}
(e_\inverter - (1+\eps) \cdot L)/H 
& = t \cdot \hlow/H - (1+\eps) \cdot s / a\\
&\geq t \cdot \hlow/\hhigh - (1+\eps) \cdot s / a\\
&\geq t \cdot (1-4/k) - (1+\eps) \cdot s / a\\
&\geq t - 4t/k - (1+\eps) \cdot s / a\\
&\geq t - 2/k - (1+\eps) \cdot s / a\\
&\geq t - 2/k - 2 \cdot (1/20k)\\
&\geq t - 3/k
\end{align*}
units of \ing. Here we used that $t \leq 1/2$, that $s \leq 1/(20 k d |V|) \leq 1/20k$, that $a \geq 1$, and that $1+\eps \leq 2$.
\end{proof}

\begin{lemma}
\label{lem:nand_gate_inverter}
For any \NAND gate, the inverter buyer's allocations of the input goods, \inoneg and \intwog, satisfy $x_{\inverter,\inoneg}, x_{\inverter,\intwog} \in [t - 5/k ,t]$. 
\end{lemma}

\begin{proof}
This proof is very similar to the proof of \cref{lem:not_gate_inverter},
but now we must account for the fact that the inverter has two input goods.

We first consider the case where $p_\outg > L$. By \cref{lem:not_gate_big_l}, the inverter must buy $t$ units of \inoneg and $t$ units of
\intwog before buying any other good. There are now two sub-cases.
\begin{itemize}
\item If $e_\inverter \ge t \cdot p_\inoneg + t \cdot p_\intwog$ then the
inverter has enough money to buy $t$ units of $\inoneg$ and $t$ units of $\intwog$, and does so.
\item If $e_\inverter < t \cdot p_\inoneg + t \cdot p_\intwog$ then the
inverter does not have enough money to buy $t$ units of \inoneg and $t$ units
of \intwog. This means that they spend all of their money on the input goods,
and therefore they they will demand zero units of \outg, which contradicts
\cref{lem:non_zero_output}.
\end{itemize}

We now consider the case where $0 < p_\outg \leq L$. First note that the
marginal utility of \outg is $s/p_\outg > s/L = a/H$, whereas the marginal
utility of \refg is $1/p_\refg = s/H$. Since $a > s$, and since the inverter is
never satiated by \outg, this means that the inverter cannot spend anything on
\refg. 
Due to the $\eps$-clearing constraint, we know that the inverter can buy
at most $1+\eps$ units of \outg, and so the inverter will spend at most
$(1+\eps) \cdot L$ money on \outg. 

All remaining money must be spent on \inoneg and \intwog so the inverter will
spend at least $e_\inverter - (1+\eps) \cdot L$ money on these two goods. Note
also that the inverter cannot buy more than~$t$ units of either input good,
since their marginal utility for that input becomes $0$ once they have~$t$ units. 
Thus, for any input good $j \in \{\inoneg, \intwog\}$, in the worst case 
the inverter buys $t$ units of the other input and pays price $H$ for those units, and then spends the rest of
their money on good $j$ and pays price $H$ for those units as well. So the inverter will buy at least 
\begin{align*}
(e_\inverter - t \cdot H - (1+\eps) \cdot L)/H 
&= 2t \cdot \hlow/H - t  - (1+\eps) \cdot s / a\\
&\geq 2t \cdot (1-4/k) - t  - 2 \cdot (1/20k) \\
&\geq t - 8t/k - 1/10k\\
&\geq t - 4/k - 1/10k\\
&\geq t - 5/k
\end{align*}
units of good $j$.
Here we used that $t \leq 1/2$, that $s \leq 1/(20 k d |V|) \leq 1/20k$, that $a \geq 1$, and that $1+\eps \leq 2$.
\end{proof}

We do not need to treat the \PURIFY gates separately, since each \PURIFY gate is
constructed entirely out of \NOT gates. Combining the previous two lemmas gives
the following bound.

\begin{lemma}
\label{lem:output_good_outside_gate}
For any gate, the total amount of the output good that is allocated to buyers that are not part of the gate's construction lies in $[2\tmin, 2t]$, where $\tmin := t - 5/k$.
\end{lemma}

\begin{proof}
By construction, for each good, there are exactly two buyers that are not part of the gate's
construction that are interested in buying that good. Those buyers are either
inverter buyers that implement a \NOT or \NAND gate, or are auxiliary buyers
that top-up the amount bought in the case where the variable encoded by that
good is the input to fewer than two other gates. By 
\cref{lem:not_gate_inverter} and
\cref{lem:nand_gate_inverter} the inverter buyers will buy some amount in the
range $[t - 5/k ,t]$, while by \cref{lem:auxiliary} the auxiliary buyers will
buy exactly $t$ units of the good. 
Thus the total amount of the good bought by these buyers lies in the range
$[2\tmin, 2t]$.
\end{proof}

\subsection{Correctness of the gates}

We now show that the gate constructions correctly simulate the gates of a
\pcircuit. In this section we will also use the notation $\lhigh := s \cdot \hhigh / a$. Since $H \in [\hlow,\hhigh]$ and $L = s \cdot H / a$, we thus also have $L \leq \lhigh$.

\paragraph{\bf \NOT gates.}

The constraints of a \NOT gate in a \pcircuit only require the gate to work
when the input is either $0$ or $1$. So to prove that the \NOT gate works, it is
sufficient to consider the cases where $p_\ing \le L$ and $p_\ing \ge H$.
However, since we use \NOT gates to implement \PURIFY gates, and since \PURIFY
gates must output a pure value even when the input is $\garbo$, we will need to
understand the relationship between the price of the output good and the price
of the input good even when $L < p_\ing < H$. We do this in the following pair
of lemmas, which give upper and lower bounds on $p_\outg$ with respect to
$p_\ing$. 

\begin{lemma}
\label{lem:not_upper_bound}
For any $\NOT(\ing, \outg, r)$ gate with $r \in [0,1]$ satisfying $1 - 2t - r - \eps > 0$, we have
\begin{equation*}
p_\outg \le \max \left(L, \; (\hlow - p_\ing) \cdot \frac{t}{1 - 2t - r - \eps}
\right).
\end{equation*}
\end{lemma}

\begin{proof}
We begin by considering the case where $p_\outg > L$. In this case, by
\cref{lem:not_gate_big_l}, which works for \NOT gates with any value of $r$,
we have that the inverter of the \NOT gate must buy $t$ units of
the input good before buying any other good. This means that the inverter can spend at
most 
\begin{align*}
& e_\inverter - t \cdot p_\ing = t \cdot (\hlow - p_\ing)
\end{align*}
money on goods other than the input good.

By \cref{lem:external_demand}, buyers other than
the inverter demand at most $2t + r$
units of the output good. Thus, to ensure that the output good $\eps$-clears, we require that at least
$1 - \eps - 2t - r$
units of the output good are bought by the inverter, and so the inverter must
spend at least $p_\outg \cdot ( 1 - \eps - 2t -r)$ money on the output good. 

For this to be possible we must have
$$t \cdot (\hlow - p_\ing) \ge p_\outg \cdot ( 1 - \eps - 2t - r),$$ 
and therefore
$$p_\outg \le (\hlow - p_\ing) \cdot \frac{t}{ 1 - \eps - 2t -r}$$
since $1 - \eps - 2t -r > 0$.
So we have shown that either $p_\outg \le L$, or that the bound above holds. So
we can conclude that 
$$p_\outg \le \max \left(L, \; (\hlow - p_\ing) \cdot \frac{t}{ 1 - \eps - 2t
-r}
\right) .$$
\end{proof}

\begin{lemma}
\label{lem:not_lower_bound}
For any $\NOT(\ing, \outg, r)$ gate with $r \in [0,1]$ satisfying $1 - 2t - r - \eps > 0$, we have
\begin{equation*}
p_\outg \ge \min \left( H, \;  (\hlow - p_\ing) \cdot \frac{t}{1 - 2\tmin - r +
\eps} \right).
\end{equation*}
\end{lemma}

\begin{proof}
We first consider the case where $p_\outg < H$. In this case, the marginal
utility of buying the reference good is $1/p_\refg = s/H$, whereas the marginal
utility of the output good is $s / p_\outg > s/H$. Thus the inverter strictly
prefers the output good to the reference good. Since the inverter can never be
satiated by the output good, this means that the inverter cannot buy the
reference good. 

Furthermore, the inverter cannot buy more than $t$ units of the input good,
because the inverter's marginal utility becomes 0 at that point. 
This means that the inverter can spend at most $t \cdot p_\ing$ money on the
input good, and will have at least
\begin{align*}
& e_\inverter - t \cdot p_\ing = t \cdot (\hlow - p_\ing)
\end{align*}
money left over after buying all goods other than \outg. This money must
therefore be spent on \outg. 

By \cref{lem:output_good_outside_gate} we have that buyers who are not part of
the \NOT gate will buy at least 2\tmin units of the output good, and by
\cref{lem:auxiliary} we have that $r$ units of the output good will be bought
by the auxiliary buyer of the \NOT gate. Since the good must $\eps$-clear, this
means that the inverter can buy at most $1 - 2\tmin - r + \eps$ units of the
output good. 

Since $t \cdot (\hlow - p_\ing)$ money must be spent on the output good, and
since at most $1 - 2\tmin - r + \eps$ units of that good can be bought, we have
$$t \cdot (\hlow - p_\ing) \le (1 - 2\tmin - r + \eps) \cdot p_\outg,$$
because otherwise there would be money left over that is not spent on any good,
which cannot happen in an optimal allocation. Rearranging this gives
$$ p_\outg \ge (\hlow - p_\ing) \cdot \frac{t}{1 - 2\tmin - r + \eps}$$
since $1 - 2\tmin - r + \eps \geq 1 - 2t - r - \eps > 0$.

So we have shown that we either have $p_\outg \ge H$, or we have the inequality
given above. Hence we have
$$ p_\outg \ge \min \left( H, \;  (\hlow - p_\ing) \cdot \frac{t}{1 -
2\tmin - r + \eps} \right).$$
\end{proof}

We can now prove that the \NOT gate works for pure values by applying the
previous two lemmas. 

\begin{lemma}
\label{lem:not_gate_correct}
For each \NOT gate with input \ing and output \outg, we have the following.
\begin{itemize}
\item If $p_\ing \ge H$ then $p_\outg \le L$.
\item If $p_\ing \le L$ then $p_\outg \ge H$.
\end{itemize}
\end{lemma}
\begin{proof}
For the first claim we can apply \cref{lem:not_upper_bound}
to obtain
\begin{align*}
p_\outg 
&\le \max \left(L, \; (\hlow - p_\ing) \cdot \frac{t}{( 1 - 2t - \rnot - \eps)} \right) \\
&\le \max \left(L, \; (\hlow - H) \cdot \frac{t}{( 1 - 2t - \rnot - \eps)} \right) \\
&\le \max \left(L, 0 \right) \\
&= L,
\end{align*}
where the second inequality uses the fact that $p_\ing \ge H$, and the third inequality used the fact that $\hlow - H < 0$, that $t > 0$,
and that $1 - \eps - 2t - \rnot = 1/11 - \eps > 0$.

For the second claim we can apply \cref{lem:not_lower_bound} to obtain
\begin{align*}
p_\outg 
&\ge \min \left( H, \;  (\hlow - p_\ing) \cdot \frac{t}{1 - 2\tmin - \rnot + \eps} \right) \\
&\ge \min \left( H, \;  (\hlow - \lhigh) \cdot \frac{t}{1 - 2\tmin - \rnot +
\eps} \right) 
\end{align*}
where the second inequality uses the fact that $p_\ing \le L \le \lhigh$.
\begin{align*}
(\hlow - \lhigh) \cdot \frac{t}{1 - 2\tmin - \rnot + \eps}
&\ge (H - 2s/k - s \cdot \hhigh / a) \cdot \frac{t}{1 - 2\tmin - \rnot + \eps} \\
&\ge (H - 2s/k - s \cdot (H + 2s/k) / a) \cdot \frac{t}{1 - 2\tmin - \rnot + \eps} \\
&\ge (H - H/k - s \cdot (H + H/k) / a) \cdot \frac{t}{1 - 2\tmin - \rnot + \eps} \\
&\ge (1-2/k-2/k^2) \cdot H \cdot \frac{t}{1 - 2\tmin - \rnot + \eps}. 
\end{align*}
The second inequality uses the fact that $\hlow \ge H - 2s/k$ which arises from
\cref{eq:high-minus-low} and that $\lhigh = s \cdot \hhigh / a$ by definition.
The third inequality uses the fact that $\hhigh \le H + 2s/k$ which again arises
from \cref{eq:high-minus-low}.
The fourth inequality uses the fact that $H = s \cdot p_\refg \le 2s$ from
\cref{lem:p_refg_bounds}.
The fifth inequality uses the fact that $s \leq 1/(20 k d |V|) \leq 1/k$ and
$a \ge 1$. 
Since $k \ge 3$, we have that
$(1-2/k-2/k^2) > 0$, and so we can continue with the following chain of
inequalities on the multiplier of $H$
\begin{align*}
\frac{t\cdot (1-2/k-2/k^2)}{1 - 2\tmin - \rnot + \eps} 
& \ge  \frac{4/11 \cdot (1-2/k-2/k^2)}{1 - 8/11+5/k - \rnot + \eps} \\
& \ge  \frac{4/11 \cdot (1-2/k-2/k^2)}{2/11+5/k} \\
& \ge 1.
\end{align*}
The first inequality used 
the fact that $\tmin \ge t - 5/k$ from \cref{lem:output_good_outside_gate}, the
fact that $t = 4/11$ by definition, the second inequality uses the fact that
$\rnot = 2/11$, and the final inequality uses the fact that $k \ge 32$.
So we have shown that $p_\outg \ge \max(H, H)$, and 
therefore we can conclude that $p_\outg \ge H$.
\end{proof}

\paragraph{\bf \NAND gates.}

For \NAND gates we simply need to verify that the gate works for pure inputs.
The following pair of lemmas shows that the \NAND gate construction correctly
implements the constraints of the \NAND gate.

\begin{lemma}
\label{lem:nand_one}
For each \NAND gate with inputs \inoneg and \intwog and output \outg, if
$p_\inoneg \ge H$ and $p_\intwog \ge H$, then $p_\outg \le L$. 
\end{lemma}

\begin{proof}
Assume, for the sake of contradiction, that 
$p_\outg > L$. Then by \cref{lem:not_gate_big_l} we have that the inverter must
buy $t$ units of \inoneg and $t$ units of \intwog before buying any other good.
The inverter's remaining budget after buying $t$ units of both of the inputs goods is
\begin{align*}
e_\inverter - t \cdot p_\inoneg - t \cdot p_\intwog
&= 2t \cdot \hlow - t \cdot p_\inoneg - t \cdot p_\intwog \\
&\le 2t \cdot \hlow - t \cdot H - t \cdot H \\
&\le 0.
\end{align*}
Hence the inverter cannot spend any money on the output good, which
contradicts \cref{lem:non_zero_output}.
\end{proof}

\begin{lemma}
\label{lem:nand_two}
For each \NAND gate with inputs \inoneg and \intwog and output \outg, if
there exists an input good $j \in \{\inoneg, \intwog\}$
such that $p_j \le L$, then $p_\outg \ge H$. 
\end{lemma}
\begin{proof}
Assume, for the sake of contradiction, that $p_\outg < H$. We have that the
marginal utility of buying the reference good is $1/p_\refg = s/H$, whereas the
marginal utility of the output good is $s / p_\outg > s/H$. Thus the inverter
strictly prefers the output good to the reference good. Since the inverter can
never be satiated by the output good, this means that the inverter cannot buy
the reference good. 

Furthermore, the inverter cannot buy more than $t$ units of either of the input
goods,
because the inverter's marginal utility becomes 0 once $t$ units have been
bought. 
This means that the inverter can spend at most $t \cdot p_\inoneg$ money on
\inoneg and $t \cdot p_\intwog$ on \intwog. The inverter's remaining budget
will therefore be at least
\begin{align*}
e_\inverter - t \cdot p_\inoneg - t \cdot p_\intwog
& = t \cdot (2\hlow - p_\inoneg - p_\intwog) \\
& \ge t \cdot (2\hlow - H - L) 
\end{align*}
money left over after buying all goods other than \outg, where the final
inequality has used the assumption that one of the two input goods $j$ satisfies $p_j
\le L$ while the other has price at most $H$ due to \cref{lem:p_j_bounds}.
This money must
therefore be spent on \outg. Therefore, the number of units of the output good bought by the inverter will be at least
\begin{align*}
t \cdot (2\hlow - H - L) / p_\outg
& \ge t \cdot (2\hlow - H - L) / H  \\
& \ge t \cdot (2\hlow - \hhigh - \lhigh) / \hhigh  \\
& \ge t \cdot (2\hlow - \hhigh - s \cdot \hhigh) / \hhigh  \\
& \ge t \cdot (2\hlow - 1.05 \hhigh) / \hhigh  \\
& \ge t \cdot (2\hlow/\hhigh - 1.05) \\
& \ge t \cdot (2 \cdot (1 - 4/k) - 1.05) \\
& > 0.75 t \\
& = 3/11.
\end{align*}
The first inequality uses the fact that $p_\outg \le H$ while the second
inequality uses the fact that $H \le \hhigh$ and $L \le \lhigh$.
The third inequality uses the fact that $\lhigh = s \cdot \hhigh / a$ by
definition, and since $a \ge 1$ we therefore have $\lhigh \le s \cdot \hhigh$.
The fourth inequality uses the fact that $s \leq 1/(20 k d |V|) \leq 1/20$.
The sixth inequality uses 
fact that $\hlow / \hhigh \ge 1 - 4/k$ from Equation \cref{eq:low-over-high},
and the final inequality uses the fact that $k > 40$ and the fact that $t =
4/11$.

By \cref{lem:output_good_outside_gate} we have that buyers who are external to
the \NAND gate will buy at least 2\tmin units of the output good, and by
\cref{lem:auxiliary} we have that $\rnand$ units of the output good will be bought
by the auxiliary buyer of the \NAND gate. Since the good must $\eps$-clear, this
means that the number of units of the output good that the inverter can buy is
at most 
\begin{align*}
1 - 2\tmin - \rnand + \eps 
& < 1 - 8/11 + 10/k - 2/11 + 1/11 \\
& = 2/11 + 10/k. \\
& < 3/11. 
\end{align*}
where the first inequality uses the fact that $\tmin \ge t - 5/k$ from 
\cref{lem:output_good_outside_gate}, that $\rnand = 2/11$, and that $\eps <
1/11$, while the second inequality uses the fact that $k > 110$.

So we have shown that the inverter must buy strictly more than 3/11 units of
the output good, but also can buy strictly less than 3/11 units of the output
good, so we have arrived at a contradiction. 
\end{proof}

\paragraph{\bf \PURIFY gates.}

We now prove the correctness of the \PURIFY gates.
For the proofs of this section, we will use the following auxiliary notation.
Let $A := \frac{1 - 2t - r - \eps}{t}$, $A' := \frac{1 - 2t - r' - \eps}{t}$, $B := \frac{1 - 2\tmin - r + \eps}{t}$, $B' := \frac{1 - 2\tmin - r' + \eps}{t}$, where $A, A', B, \in (0,1)$, and $0 < B' = 1 - 4/k - s/a < 1$.

In the following two lemmas and corollary, we prove a relationship between the
input price to a chain of \NOT gates and the output price. 

\begin{lemma}\label{lem:lower_bound_NOT_chain}
    Consider a chain of gates $\NOT(\ing, \midg, r)$, $\NOT(\midg, \outg, r')$. For any $P \in [L,H]$, if $p_{\ing} \geq \hlow \cdot \left( 1 - A \right) + P \cdot AB'$ then $p_{\outg} \geq P$.
\end{lemma}

\begin{proof}
     Consider the gate $\NOT(\ing, \midg, r)$ with $p_{\ing} \geq \hlow \cdot \left( 1 - A \right) + P \cdot AB'$. This implies that $\hlow - p_{\ing} \leq \hlow \cdot A - P \cdot AB'$, or equivalently, $ (\hlow - p_{\ing})/A \leq \hlow - P \cdot B'$.
     Now notice that $\hlow - P \cdot B' \geq \hlow - P \cdot \left( \frac{\hlow}{\hhigh} - L/H \right) \geq L$,
     where the two inequalities come from the fact that $ B' = 1 - 4/k - s/a \leq \frac{\hlow}{\hhigh} - \frac{L}{H}$ and $P \leq H \leq \hhigh$, respectively.
     Therefore, \cref{lem:not_upper_bound} implies that $p_\midg \leq \hlow - P \cdot B'$.

    Now consider the gate $\NOT(\midg, \outg, r')$. From above, we have $\hlow - p_\midg \geq P \cdot B'$, which implies $(\hlow - p_\midg)/B' \geq P$. Since $P \leq H$, \cref{lem:not_lower_bound} implies that $p_\outg \geq P$.
\end{proof}

\begin{lemma}\label{lem:upper_bound_NOT_chain}
    Consider a chain of gates $\NOT(\ing, \midg, r)$, $\NOT(\midg, \outg, r')$. For any $P \in [L,H]$, if $p_{\ing} \leq \hlow \cdot \left( 1 - B \right) +  P \cdot A'B$ then $p_{\outg} \leq P$.
\end{lemma}

\begin{proof}
     Consider the gate $\NOT(\ing, \midg, r)$ with $p_{\ing} \leq \hlow \left( 1 - B \right) +  P \cdot A'B$. 
     This implies that $\hlow - p_{\ing} \geq \hlow \cdot B - P \cdot A'B$, or equivalently, $ (\hlow - p_{\ing})/B \geq \hlow - P \cdot A'$.
     Notice that $\hlow - P \cdot A' \leq \hlow \leq H$.
     Therefore, \cref{lem:not_lower_bound} implies that $p_\midg \geq \hlow - P \cdot A'$.

    Now consider the gate $\NOT(\midg, \outg, r')$. As shown above, we have $\hlow - p_\midg \leq P \cdot A'$, which means that $(\hlow - p_\midg)/A' \leq P$. Since $P \geq L$, \cref{lem:not_upper_bound} implies that $p_\outg \leq P$.
\end{proof}

\begin{corollary}\label{cor:garbage_region}
    For some even $d \geq 2$, consider a chain made by connecting $d/2$ pairs of gates $\NOT(g_{j-1}, g_{j}, r)$, $\NOT(g_{j}, g_{j+1}, r')$, where $j \in \{ 1, 3, 5, \dots, d-1 \}$. Also, let $\ing := g_0$, and $\outg := g_d$ be the input and output of the chain, respectively.
    \begin{itemize}
        \item If $p_{\ing} \geq \hlow \cdot \frac{1-A}{1-AB'} + (AB')^{d/2} \cdot \left( H - \hlow \cdot \frac{1-A}{1-AB'} \right)$, then $p_{\outg} \geq H$, and
        \item if $p_{\ing} \leq \hlow \cdot \frac{1-B}{1-A'B} + (A'B)^{d/2} \cdot \left( L - \hlow \cdot \frac{1-B}{1-A'B} \right)$, then $p_{\outg} \leq L$.
    \end{itemize}
\end{corollary}

\begin{proof}
    By construction, the last pair's input serves as the output of the previous pair, and so on, until the second pair's input serves as the output of the first pair. Suppose now that we require the output of the chain to be at least $P \in [L,H]$. Then by repeatedly applying the bound of \cref{lem:lower_bound_NOT_chain}, i.e., substituting $P$ with $\hlow \cdot \left( 1 - A \right) + P \cdot AB'$ for $d/2$ times, and finally setting $P = H$, we get the first part of the statement. 
    
    In particular, to prove the first part of the statement, let the lower bound for the input of the $(d/2-i+1)$-st pair of \NOT gates be $P_i$, for $i \in \{ 1,\dots, d/2\}$. Now set $P_0 = H$, and observe that the first time we substitute (i.e., we consider only the last pair of gates), the lower bound for $p_\ing$ becomes $P_1 := \hlow \cdot \left( 1 - A \right) + P_0 \cdot AB'$ which agrees with the first part of the statement for $d=2$. We will use this as a base case. Now, given that for some $i \in \{ 1, \dots, d/2-1 \}$ it is $P_i := \hlow \cdot \frac{1-A}{1-AB'} + (AB')^i \cdot \left( P_0 - \hlow \cdot \frac{1-A}{1-AB'} \right)$, we will prove that if $P_{i+1} = \hlow \cdot \left( 1 - A \right) + P_{i} \cdot AB'$, then $P_{i+1} = \hlow \cdot \frac{1-A}{1-AB'} + (AB')^{i+1} \cdot \left( P_0 - \hlow \cdot \frac{1-A}{1-AB'} \right)$. Indeed, by this substitution we get
    \begin{align*}
        P_{i+1} &= \hlow \cdot \left( 1 - A \right) + \left( \hlow \cdot \frac{1-A}{1-AB'} + (AB')^i \cdot \left( P_0 - \hlow \cdot \frac{1-A}{1-AB'} \right) \right) \cdot AB' \\
        &= \hlow \cdot \left( 1 - A \right) \left( 1 + \frac{AB'}{1-AB'} \right) + (AB')^i \cdot \left( P_0 - \hlow \cdot \frac{1-A}{1-AB'} \right)  \cdot AB' \\
        &= \hlow \cdot \frac{1 - A}{1-AB'} + (AB')^{i+1} \cdot \left( P_0 - \hlow \cdot \frac{1-A}{1-AB'} \right) ,
    \end{align*}
    therefore, by induction on $i$ until $i=d/2$, we complete the first statement of the lemma, for $P_0 = H$.

    The second part of the proof is symmetric, using \cref{lem:upper_bound_NOT_chain} and $P_0 = L$, and we omit it.
\end{proof}

We now prove that the \PURIFY gate is correct.
Specifically, we consider a \PURIFY gate with an input good $\ing$ and two output goods $\outg 1$, $\outg 2$. It consists of two chains of \NOT gates; one, called \textit{chain 1}, with input $\ing$ and output $\outg 1$, and the other, called \textit{chain 2}, with input $\ing$ and output $\outg 2$. Each chain contains $d/2$ pairs of \NOT gates, where $d \geq 2 \ceil{\log_{2}(3/\delta)}$ is an even number. Each chain $i \in \{1,2\}$ has pairs of gates $\NOT(g^i_{j-1}, g^i_{j}, r^i_{j})$, $\NOT(g^i_{j}, g^i_{j+1}, r^i_{j+1})$, where the $\pairs$-th pair, $\pairs \in \{ 1, \dots, d/2 \}$ corresponds to the aforementioned $j$-th and $(j+1)$-st \NOT gates, with $j = 2\pairs-1$. Also, the first good of both chains is common, i.e., $g^1_0 = g^2_0$. Finally, chain 1 has $\ing := g^1_0$, $\outg 1 := g^1_d$, $r^1_j = 0$ for $j$ odd, and $r^1_{j} = 2/11$ for $j$ even; chain 2 has $\ing := g^2_0$, $\outg 2 := g^2_d$, $r^2_j = 2/11$ for $j$ odd, and $r^2_{j} = 0$ for $j$ even.
 
\begin{lemma}
\label{lem:purify_one}
    For each \PURIFY gate we have the following.
    \begin{itemize}
        \item If $p_\ing \geq H$, then $p_{\outg 1} \geq H$ and $p_{\outg 2} \geq H$.
        \item If $p_\ing \leq L$, then $p_{\outg 1} \leq L$ and $p_{\outg 2} \leq L$.
        \item If $p_\ing \in (L,H)$, then at least one of $p_{\outg 1}, p_{\outg 2}$ is outside of $(L,H)$.
    \end{itemize}
\end{lemma}

\begin{proof}
    First, we will need to calculate the quantities $A,B,A',B'$ for each chain, where we add a subscript $i \in \{1,2\}$ to indicate the chain they refer to. For ease of presentation, we define $\delta := \frac{11}{4} \cdot ( 1/11 - \eps ) > 0$, and notice that $\tmin \geq t - 5/k \geq t-\delta/22$. By construction, for chain 1, without loss of generality we can use only the first pair of \NOT gates (i.e., $j \in \{ 1, 2 \}$), and by the specified $r^i_{j}$ above, we get
     \begin{align*}
        A_1 &:= \frac{1-2t-r^1_1-\eps}{t} = \frac{3/11 - \eps}{4/11} \\
        B_1 &:= \frac{1-2\tmin-r^1_1+\eps}{t} \in \left[\frac{3/11 + \eps}{4/11}, \frac{3/11 + \eps}{4/11} + \delta  \right]  \\
        A'_1 &:= \frac{1-2t-r^1_2-\eps}{t} = \frac{1/11 - \eps}{4/11} \\
        B'_1 &:= \frac{1-2\tmin-r^1_2+\eps}{t} \in \left[\frac{1/11 + \eps}{4/11}, \frac{1/11 + \eps}{4/11} + \delta  \right]
    \end{align*}
    Similarly, for chain 2 we have 
    \begin{align*}
        A_2 &:= \frac{1-2t-r^2_1-\eps}{t} = \frac{1/11 - \eps}{4/11} \\
        B_2 &:= \frac{1-2\tmin-r^2_1+\eps}{t} \in \left[\frac{1/11 + \eps}{4/11}, \frac{1/11 + \eps}{4/11} + \delta  \right] \\
        A'_2 &:= \frac{1-2t-r^2_2-\eps}{t} = \frac{3/11 - \eps}{4/11} \\
        B'_2 &:= \frac{1-2\tmin-r^2_2+\eps}{t} \in \left[\frac{3/11 + \eps}{4/11}, \frac{3/11 + \eps}{4/11} + \delta  \right] \\
    \end{align*}
    Finally, notice that
        $B_2 \in [A_1 - 2\delta , A_1 -\delta]$, and $B'_1 \in [A'_2 - 2\delta , A'_2 -\delta]$.
    
    \cref{cor:garbage_region} implies that for chain 1 there is an interval $(R^L_1, R^U_1)$ of $p_\ing$ prices for which $p_\outg$ is not guaranteed to be at most $L$ or at least $H$. In particular, 
    \begin{align*}
        R^L_1 := \hlow \cdot \frac{1-B_1}{1-A'_1 B_1} + (A'_1 B_1)^{d/2} \cdot \left( L - \hlow \cdot \frac{1-B_1}{1-A'_1 B_1} \right),
    \end{align*}
    and 
    \begin{align*}
        R^U_1 := \hlow \cdot \frac{1-A_1}{1-A_1 B'_1} + (A_1 B'_1)^{d/2} \cdot \left( H - \hlow \cdot \frac{1-A_1}{1-A_1 B'_1} \right)
    \end{align*}
   Similarly, the corollary implies that, for chain 2, the respective interval is $(R^L_2, R^U_2)$ with
    \begin{align*}
        R^L_2 := \hlow \cdot \frac{1-B_2}{1-A'_2 B_2} + (A'_2 B_2)^{d/2} \cdot \left( L - \hlow \cdot \frac{1-B_2}{1-A'_2 B_2} \right),
    \end{align*}
    and 
    \begin{align*}
        R^U_2 := \hlow \cdot \frac{1-A_2}{1-A_2 B'_2} + (A_2 B'_2)^{d/2} \cdot \left( H - \hlow \cdot \frac{1-A_2}{1-A_2 B'_2} \right)
    \end{align*}

    We will now show that 
    \begin{align*}
        R^L_1 < R^U_1 \leq R^L_2 < R^U_2.
    \end{align*}
    First, by \cref{cor:garbage_region}, it is immediate that $R^L_1 < R^U_1$. Otherwise, $p_\ing = R^L_1$ would result into $p_{\outg 1} \leq L$, but because $R^L_1 \geq R^U_1$, we would have $p_{\outg 1} \geq H$. This implies $H \leq L = s \cdot H/a$, a contradiction, since $s < 1/2 < a$. Similarly, by \cref{cor:garbage_region}, it is straightforward that $R^L_2 < R^U_2$. 
    Now we will show that $R^U_1 \leq R^L_2$. We have

    \begin{align}
        R^U_1 &= \hlow \cdot \frac{1-A_1}{1-A_1 B'_1} + (A_1 B'_1)^{d/2} \cdot \left( H - \hlow \cdot \frac{1-A_1}{1-A_1 B'_1} \right) \nonumber \\
        &\leq \hlow \cdot \frac{1-B_2-\delta}{1 - A'_2 B_2} + (A'_2 B_2)^{d/2} \cdot \left( H - \hlow \cdot \frac{1-B_2-2\delta}{1-A'_2 B_2} \right) \nonumber \\
        & \qquad \qquad \qquad \qquad \qquad \qquad \qquad  \text{($A_1 \in [B_2 + \delta, B_2 + 2\delta ]$, $A_1 = A'_2 $, and $B'_1 = B_2$, by definition)}  \nonumber \\
        &= \hlow \cdot \frac{1-B_2}{1 - A'_2 B_2} + (A'_2 B_2)^{d/2} \cdot \left( L - \hlow \cdot \frac{1-B_2}{1-A'_2 B_2} \right) - \hlow \cdot \frac{\delta}{1 - A'_2 B_2}  \nonumber \\ 
        & \qquad \qquad \qquad \qquad \qquad \qquad \qquad \qquad \qquad \qquad \quad + (A'_2 B_2)^{d/2} \cdot \left( H - L + \hlow \cdot \frac{2\delta}{1-A'_2 B_2} \right) \nonumber \\
        &= R^L_2 - \hlow \cdot \frac{\delta}{1 - A'_2 B_2} + (A'_2 B_2)^{d/2} \cdot \left( H - L + \hlow \cdot \frac{2\delta}{1-A'_2 B_2} \right) \qquad \quad  \text{(by definition of $R^L_2$)} \nonumber \\
        &< R^L_2 - \hlow \cdot \delta + (A'_2 B_2)^{d/2} \cdot \left( 2\hlow + \hlow \right) \quad  \left( \text{$L \geq 0$, $0 \leq A'_2 B_2 < \frac{1}{2}$, $H \leq 2 \hlow$, and $\delta \leq \frac{1}{4}$} \right)  \nonumber \\
        &< R^L_2 - \hlow \cdot \left[\delta - 3 \cdot \left( \frac{1}{2} \right)^{d/2} \right] \qquad \qquad \qquad \qquad \qquad \qquad \qquad \quad  \left( \text{$0 \leq A'_2 B_2 \leq \frac{3}{4} \cdot \frac{1}{2} < \frac{1}{2}$} \right)  \nonumber \\
        &\leq R^L_2 \qquad \qquad \qquad \qquad \qquad \qquad \qquad \qquad \qquad \qquad \left( \text{ by our choice of $d \geq 2 \ceil{\log_{2}\left(\frac{3}{\delta}\right)}$ } \right). \nonumber
    \end{align}
    
    Now we are ready to prove the first part of the lemma's statement. Suppose $p_\ing \geq H$. Notice that 
    \begin{align*}
        H &= H \left(1 - (A_2 B'_2)^{d/2} \right) + H \cdot (A_2 B'_2)^{d/2} \\
        &\geq H \cdot \frac{1-A_2}{1-A_2 B'_2} \left(1 - (A_2 B'_2)^{d/2} \right) + H \cdot (A_2 B'_2)^{d/2} \qquad \qquad \qquad \qquad \qquad \qquad \text{($0 < A_2, B'_2 < 1$)} \\
        &\geq \hlow \cdot \frac{1-A_2}{1-A_2 B'_2} \left(1 - (A_2 B'_2)^{d/2} \right) + H \cdot (A_2 B'_2)^{d/2} \\
        &= \hlow \cdot \frac{1-A_2}{1-A_2 B'_2} + (A_2 B'_2)^{d/2} \cdot \left( H - \hlow \cdot \frac{1-A_2}{1-A_2 B'_2} \right) \\
        &= R^U_2 .
    \end{align*}
    Therefore, $p_\ing \geq R^U_2 > R^U_1$, and \cref{cor:garbage_region} implies that $p_{\outg 1} \geq H$, and $p_{\outg 2} \geq H$.

    Moving on to the proof of the second part of the statement, suppose $p_\ing \leq L$. We have 
    \begin{align*}
        L &= L \left(1 - (A'_1 B_1)^{d/2} \right) + L \cdot (A'_1 B_1)^{d/2} \\
        &\leq \frac{s\delta}{2} \cdot \left(1 - (A'_1 B_1)^{d/2} \right) + L \cdot (A'_1 B_1)^{d/2} \qquad \qquad \qquad \qquad \qquad \qquad\qquad \qquad \qquad \text{($L \leq s\delta/2 $)} \\
        &\leq \hlow \cdot \frac{1-B_1}{1-A'_1 B_1} \left(1 - (A'_1 B_1)^{d/2} \right) + L \cdot (A'_1 B_1)^{d/2} \quad  \text{($ \hlow \geq \hmin \geq s/2, \delta \leq 1, 0 < A'_1, B_1 < 1 $)} \\
        &= \hlow \cdot \frac{1-B_1}{1-A'_1 B_1} + (A'_1 B_1)^{d/2} \cdot \left( L - \hlow \cdot \frac{1-B_1}{1-A'_1 B_1} \right) \\
        &= R^L_1 .
    \end{align*}
    Therefore, $p_\ing \leq R^L_1 < R^L_2$, and so, \cref{cor:garbage_region} implies that $p_{\outg 1} \leq L$, and $p_{\outg 2} \leq L$.

    Now, for the third part of the statement, suppose that $p_{\ing} \in (L,H)$. If $p_{\ing} \leq R^L_2$, then by \cref{cor:garbage_region} we have that $p_{\outg 2} \leq L$. If $p_{\ing} > R^L_2$, then $p_{\ing} > R^U_1$ (since we showed that $R^L_2 \geq R^U_1$), and again from \cref{cor:garbage_region}, $p_{\outg 1} \geq H$. Therefore, if $p_\ing \in (L,H)$, then it cannot be that both $p_{\outg 1}$ and $p_{\outg 2}$ are in $(L,H)$.
\end{proof}

\section{Inapproximability for Arrow-Debreu Exchange Markets}

In this section we show that our inapproximability result also applies to Arrow-Debreu exchange markets.

\paragraph{\bf Exchange markets.} An exchange market is given by a tuple $\big(G,B,(w_{i,j})_{i\in B, j \in G},(u_i)_{i \in B}\big)$, where:
\begin{itemize}
\item $G$ is a set of (divisible) goods.
\item $B$ is a set of buyers (or traders).
\item For every $i \in B$ and $j \in G$, $w_{i,j} \geq 0$ is the endowment of good $j$ owned by buyer $i$. Without loss of generality, we assume that there is one unit of each good available, i.e., $\sum_{i \in B} w_{i,j} = 1$ for all $j \in G$.
\item For every $i \in B$, $u_i: \mathbb{R}_{\geq 0}^{|G|} \to \mathbb{R}_{\geq 0}$ is an SPLC utility function, as defined in the preliminaries.
\end{itemize}

\paragraph{\bf Optimal bundles.}
Given a price vector $p \in \mathbb{R}_{\geq 0}^{|G|}$, the set of optimal bundles for buyer $i$, denoted $\opt_i(p) \subseteq \mathbb{R}_{\geq 0}^{|G|}$, is the set of optimal solutions of the following optimization problem:
\begin{equation*}
\begin{split}
\max \quad &u_i(x_i) \\
\text{ s.t.} \quad  
& \sum_{j \in G} p_j x_{i,j} \leq \sum_{j \in G} p_j w_{i,j} \\
& x_{i,j} \geq 0 \quad \forall j \in G.
\end{split}
\end{equation*}
In other words, the budget of the buyer is the amount of money obtained by selling its endowment.

\paragraph{\bf Existence of equilibria.}
The definition of market equilibrium is identical to the one given for Fisher markets in the preliminaries.\footnote{This is sometimes called an $\eps$-tight market equilibrium; see the discussion in \cite{ChenPY17-non-monotone-markets}.} The following condition is sufficient to guarantee the existence of a market equilibrium for exchange markets \cite{Maxfield1997,vazirani2011market-plc-ppad}:

\begin{description}
    \item[Sufficient Condition:] The economy graph of the market is strongly connected. This graph is defined on the set of buyers $B$ by introducing a directed edge from buyer $i$ to buyer $i'$ if there exists a good $j \in G$ such that $w_{i,j} > 0$ and $u_{i',j}$ is a strictly increasing function.
\end{description}

\begin{theorem}
It is \ppad/-complete to compute an $\eps$-approximate market equilibrium in Arrow-Debreu exchange markets with SPLC utilities for any constant $\eps < 1/11$.
\end{theorem}

\begin{proof}
To prove this result we use a simple folklore reduction from Fisher markets to exchange markets. Fix any $\eps < 1/11$ and let $(G,B,(e_i)_{i\in B},(u_i)_{i \in B})$ denote a Fisher market that satisfies the sufficient condition for Fisher markets. We construct the corresponding exchange market $(G,B,(w_{i,j})_{i\in B, j \in G},(u_i)_{i \in B})$, where the endowments are given by
$$w_{i,j} := \frac{e_i}{\sum_{k \in B} e_k}$$
for all $i \in B$ and $j \in G$. Note that have $\sum_{i \in B} w_{i,j} = 1$ for all $j \in G$, as desired. Furthermore, it is easy to see that the exchange market satisfies the sufficient condition, because we have $w_{i,j} > 0$ for all $i,j$ and moreover, by the sufficient condition for the Fisher market, for any $i \in B$ there exists a good $j \in G$ such that $u_{i,j}$ is a strictly increasing function. Thus, the economy graph is strongly connected.

Now consider any $\eps$-approximate market equilibrium $(p,x)$ of the exchange market. It is easy to see that the prices are invariant to scaling, so without loss of generality we can assume that
$$\sum_{j \in G} p_j = \sum_{i \in B} e_i.$$
As a result the budget available to buyer $i$ in the exchange market at prices $p$ is
$$\sum_{j \in G} p_j w_{i,j} = \sum_{j \in G} p_j \frac{e_i}{\sum_{k \in B} e_k} = e_i \frac{\sum_{j \in G} p_j}{\sum_{k \in B} e_k} = e_i.$$
But this is exactly the budget of buyer $i$ in the Fisher market, so it follows that bundle $x_i$ is also optimal for buyer $i$ in the Fisher market. Finally, since $\sum_{i \in B} w_{i,j} = 1$ for all $j \in G$, the $\eps$-clearing constraint in the exchange market and Fisher market are identical. It follows that $(p,x)$ is also an $\eps$-approximate market equilibrium for the Fisher market. Thus, finding an $\eps$-approximate market equilibrium $(p,x)$ of the exchange market is \ppad/-hard. Furthermore, membership in \ppad/ is known from \cite{vazirani2011market-plc-ppad}.
\end{proof}

\begin{acks}
Argyrios Deligkas was supported by EPSRC Grant EP/X039862/1 ``NAfANE: New Approaches for Approximate Nash Equilibria''.
John Fearnley was supported by EPSRC Grant EP/W014750/1
``New Techniques for Resolving Boundary Problems in Total Search''.
\end{acks}

\bibliographystyle{ACM-Reference-Format}
\bibliography{references}

\end{document}